\documentclass[format=acmsmall, review=false]{acmart}
\usepackage{acm-ec-24}
\usepackage{booktabs} % For formal tables
\usepackage[ruled]{algorithm2e} % For algorithms
\usepackage{colortbl}

\SetAlFnt{\small}
\SetAlCapFnt{\small}
\SetAlCapNameFnt{\small}
\SetAlCapHSkip{0pt}
\IncMargin{-\parindent}
\usepackage[capitalize]{cleveref}

\usepackage{dsfont}
\newcommand{\bone}{\mathds{1}}
\definecolor{color2}{RGB}{162,173,0}
\definecolor{color3}{RGB}{227,114,34}
\definecolor{color4}{RGB}{202,033,063}
\usepackage{tikz}
\usetikzlibrary{patterns}
\usepackage{calc}
\usepackage{thmtools}
\usepackage{enumitem} % reduce list vspace
\usepackage{soul} %strikethrough
\usepackage{thm-restate}
\setlist{parsep = 0.2\parsep, itemsep = 2\itemsep, leftmargin=\parindent}%,labelindent=0}

\setcitestyle{acmnumeric}
\newif\ifcomment 
\commenttrue

\usepackage{amsmath}
\usepackage{graphicx}
\usepackage{float}
\usepackage{caption}
\usepackage{amsthm}
\usepackage{color}
\definecolor{green}{rgb}{0,0.5977,0}

\newcommand{\rr}{\mathbb R}

\newcommand{\abs}[1]{\left|{#1}\right|}

\newcommand{\suchthat}{\ | \ }

\newcommand{\tth}{^\text{th}}

\newcommand{\txt}[1]{\text{#1}}

\newcommand{\push}{\\ & \ \ \ \ \ \ \ \ \ \ }

\makeatletter 
\g@addto@macro{\@algocf@init}{\SetKwInOut{Parameter}{Parameters}} 
\makeatother

\newcommand{\pdif}[1]{\tfrac{\partial}{\partial#1}}
\newcommand{\pips}{$\pi$ps}
\renewcommand{\Pr}{\mathbb{P}} %\mathop{\mathbb{P}}}
\newcommand{\app}{a}
\renewcommand{\theta}{\vartheta}
\newcommand{\introcite}[1]{\citeauthor{#1} (\citeyear{#1})}

\theoremstyle{acmplain}
\newtheorem{observation}{Observation}[section]
\title{Monotone Randomized Apportionment}

\makeatletter
\def\@ACM@checkaffil{% Only warnings
\if@ACM@instpresent\else
\ClassWarningNoLine{\@classname}{No institution present for an affiliation}%
\fi
\if@ACM@citypresent\else
\ClassWarningNoLine{\@classname}{No city present for an affiliation}%
\fi
\if@ACM@countrypresent\else
\ClassWarningNoLine{\@classname}{No country present for an affiliation}%
\fi
}
\makeatother

\author{Jos\'e Correa}
\affiliation{
\institution{Universidad de Chile}
}
\author{Paul G\"olz}
\affiliation{
\institution{UC Berkeley, Cornell University}
}
\author{Ulrike Schmidt-Kraepelin}
\affiliation{
\institution{TU Eindhoven}
}
\author{Jamie Tucker-Foltz}
\affiliation{
\institution{Harvard University}
}
\author{Victor Verdugo}
\affiliation{
\institution{Pontificia Universidad Cat\'olica de  Chile}
}

\begin{abstract}
	Apportionment is the act of distributing the seats of a legislature among political parties (or states) in proportion to their vote shares (or populations). A famous impossibility by \introcite{BY01} shows that no apportionment method can be proportional up to one seat (\emph{quota}) while also responding monotonically to changes in the votes (\emph{population monotonicity}).
	\introcite{Grimmett04} proposed to overcome this impossibility by randomizing the apportionment, which can achieve quota as well as perfect proportionality and monotonicity\,---\,at least in terms of the \emph{expected number} of seats awarded to each party.
	Still, the correlations between the seats awarded to different parties may exhibit bizarre non-monotonicities.
	When parties or voters care about joint events, such as whether a coalition of parties reaches a majority, these non-monotonicities can cause paradoxes, including incentives for strategic voting.
	
	In this paper, we propose monotonicity axioms ruling out these paradoxes, and study which of them can be satisfied jointly with Grimmett's axioms.
	Essentially, we require that, if a set of parties all receive more votes, the probability of those parties jointly receiving more seats should increase.
	Our work draws on a rich literature on \emph{unequal probability sampling} in statistics (studied as \emph{dependent randomized rounding} in computer science).
	Our main result shows that a sampling scheme due to \introcite{Sampford67} satisfies Grimmett's axioms and a notion of higher-order correlation monotonicity.
\end{abstract}

\begin{document}
	
	\begin{titlepage}
		
		\maketitle
		
	\end{titlepage}
	
	\section{Introduction}\label{secIntro}
	Across modern democracies, a long-lasting ideal has been that the legislature should ``be an exact portrait, in miniature, of the people'' it represents, in the words of John Adams~\cite{Adams76}.
	Many democratic systems aim to achieve this maxim by partitioning the seats of their legislature over political blocs \emph{in proportion} to the sizes of these blocs.
	Which blocs the countries consider differs: federal states divide up seats between the member states (in proportion to state populations), whereas 85 countries around the world divide up seats between political parties (in proportion to their vote shares)~\cite{ACE}.
	Mathematically speaking, however, both settings pose the same task of \emph{apportionment}: dividing a fixed number $h \in \mathbb{N}$ of seats across $n$ blocs in proportion to the blocs' sizes.
	For the sake of exposition, we use the language of apportionment over parties throughout this paper, but our work equally applies to other apportionment settings.
	
	Though it sounds almost trivial, the question of how to proportionally apportion seats possesses surprising mathematical and political complexity.
	The cause of this complexity is the indivisibility of seats\,---\,if a party receives $8.4\%$ of the votes and hence deserves 8.4 out of $h=100$ seats, should it receive 8 or 9, or even some other number?
	Since the 18th century, political luminaries and mathematicians have devised various \emph{apportionment methods}, i.e., functions that take in the number of votes received by all parties and the legislature size $h$, and return how many seats are to be filled by each party.
	Which of these methods should be used has led to intense debate, which often revolved around the methods' mathematical properties and actively involved mathematicians~\cite{Szpiro10}.
	
	Take, for example, \emph{Hamilton's method}, the apportionment method put forward by the first US secretary of the treasury, Alexander Hamilton.
	His method first calculates, for each party, the number of seats it proportionally deserves (8.4, in the example above) and immediately assigns the floor of this number to the party (e.g., $\lfloor 8.4 \rfloor = 8$), which is known as the party's \emph{lower quota}.
	Then, the method goes through the parties in decreasing order of their \emph{residue} (e.g., $8.4 - \lfloor 8.4 \rfloor = 0.4$) and assigns one more seat to each party until the desired house size $h$ is reached.
	One strength of Hamilton's method is that it satisfies \emph{quota}: a party's number of seats will never be below the floor of its proportional entitlement or above its ceiling.
	On the flip side, this method exhibits paradoxical non-monotonicities when the votes change.
	For instance, it violates \emph{population monotonicity}, which means that, when party~1 gains voters and party~2 loses voters from one election to the next, party~1 might lose a seat while party~2 gains one.
	Such monotonicity violations are not mere mathematical inconveniences but led to Hamilton's method being abandoned for the apportionment over US states in 1901~\cite[p.\ 42]{BY01} and for the apportionment over parties in Germany in 2008~\cite{Bgbl08}.
	In fact, no apportionment method is immune to such criticism: as Balinski and Young showed, \emph{any} method violates quota or population monotonicity~\cite{BY01}.
	
	In 2004, the mathematician Geoffrey Grimmett proposed a simple solution to this impasse: allowing apportionment methods to be random~\cite{Grimmett04}.
	Following his proposal, a party deserving 8.4 seats would be ``rounded down'' to 8 seats with 60\% probability and ``rounded up'' to 9 seats with 40\% probability. 
	In general, Grimmett's random apportionment of seats would always satisfy quota ex post and \emph{ex-ante proportionality}, which means that a party's \emph{expected number} of seats is perfectly equal to its proportional share.
	This latter property implies that the expected number of seats satisfies all reasonable forms of monotonicity, circumventing the Balinski--Young impossibility.
	
	\looseness=-1
	In designing such a randomized method, one must carefully correlate the rounding decisions across parties to ensure that the total number of seats is indeed $h$.
	Grimmett's method, like Hamilton's, 
	first awards each party its lower quota, and then (randomly) awards some parties one more seat, as follows:
	Let the parties' residues be $p_1, \dots, p_n \in [0, 1)$, and let $k$ denote their sum, which is the number of seats not yet awarded.
	We line up $n$ intervals on the number line without gaps, where the $i\tth$ interval has length $p_i$.
	We then shift all intervals to the right by a random amount $u$, drawn uniformly from $[0, 1)$, and round up exactly those parties whose interval contains an integer.
	Note that there are exactly $k$ such integers in the interval $[u, k + u)$, so this step will add $k$ seats, as required.
	Moreover, the process satisfies ex-ante proportionality since, no matter where some party $i$'s interval is placed, the random shift $u$ will leave it containing an integer with probability exactly~$p_i$.
	
	Beyond Grimmett's method, the recipe ``award lower quotas then round'' opens up a vast space of randomized apportionment methods, all of which satisfy quota and ex-ante proportionality.
	To construct such a method, all we need is a \emph{rounding rule}, i.e., a procedure that takes in residues $p_1, \dots, p_n \in [0, 1)$ summing to some integer $k$, and randomly selects a subset $S \subseteq \{1, \dots, n\}$ of size~$k$ such that each $1 \leq i \leq k$ is contained in $S$ with probability $p_i$.
	Fortuitously, this very setting has been studied in mathematical statistics as \emph{\pips{} (``probability proportional to size'') sampling without replacement}~\cite{BH83,Tille06}.
	In fact, the seminal book by \citet{BH83} lists 50 such methods,\footnote{Though several only apply if $k=2$ or give inclusion probabilities only approximately equal to the $p_i$.} out of which one, \emph{systematic sampling}~\cite{Madow49}, yields exactly Grimmett's method when applied as above.
	\citet{DT98} and \citet{HLS89} even developed general frameworks for \pips{} sampling, which capture infinitely many rounding rules and hence randomized apportionment methods.
	In the past two decades, theoretical computer scientists have searched for the same kind of rounding rules under the name of \emph{dependent randomized rounding}~\cite{GKP+06}, which have proved powerful tools in algorithm design~\cite[e.g.,][]{CL12,Srinivasan01,NSW23a}\,---\,frequently reinventing rules previously known in statistics~\cite{Tille23}.
	
	But which randomized apportionment method out of this multitude should one choose, given that all these methods satisfy quota and ex-ante proportionality?\footnote{Grimmett himself offers ``no justification for [his] scheme apart from fairness and ease of implementation''~\cite{Grimmett04}.}
	Naturally, we define additional desirable axioms and focus 
	on
	those methods with good axiomatic properties.
	In this paper, we ask
	\begin{quote}
		\emph{Which monotonicity axioms are desirable in randomized apportionment? And can sampling schemes from the statistics literature help us satisfy them?}
	\end{quote}
	
	\subsection{A Motivating Example}
	\label{sec:motivating}
	We motivate our monotonicity axioms with the story of two elections on the island of Apportia.
	Despite its modest population of 1100 people, Apportia boasts as many as six political parties, three leaning left and three leaning right.
	Apportia's parliament consists of just 11 legislators, elected at-large using Grimmett's randomized apportionment method.
	In the previous election, votes were cast according to the middle columns in \cref{tbl:apportia}.
	\definecolor{up}{rgb}{0., .5, 0.}
	\definecolor{down}{rgb}{8., 0., 0.}
	\begin{table}[htb]
		\centering
		\caption{Vote totals for six parties in two fictional elections.}\label{tbl:apportia}
		\begin{tabular}{rl rrr rrr}
			\toprule
			& & \multicolumn{3}{c}{Previous election} & \multicolumn{3}{c}{New election} \\
			\cmidrule(lr){3-5} \cmidrule(l){6-8}
			Party & Ideology & Votes & Lower quota & Residue & Votes & Lower quota & Residue \\
			\midrule
			1 & left &     110 & 1 & 0.1 &       110 & 1 & 0.1 \\
			2 & right &    270 & 2 & 0.7 &       (\textcolor{up}{+20}) 290 & 2 & 0.9 \\
			3 & left &     210 & 2 & 0.1 &       210 & 2 & 0.1 \\
			4 & right &    160 & 1 & 0.6 &       (\textcolor{up}{+30}) 190 & 1 & 0.9 \\
			5 & left &      70 & 0 & 0.7 &       (\textcolor{down}{-60}) \hphantom{0}10 & 0 & 0.1 \\
			6 & right &    280 & 2 & 0.8 &       (\textcolor{up}{+10}) 290 & 2 & 0.9 \\
			\bottomrule
		\end{tabular}
	\end{table}

	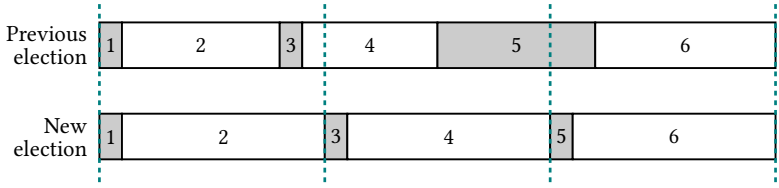
\begin{figure}[htb]
		\centering
		\scalebox{0.6}{
			\begin{tikzpicture}
				\node [align=right,left] at (-.15,3.25) {\LARGE Previous};
				\node [align=right,left] at (-.15,2.75) {\LARGE election};
				\node [align=right,left] at (-.15,1.25) {\LARGE New};
				\node [align=right,left] at (-.15,0.75) {\LARGE election};
				\foreach \a/\b/\c/\d in {0/.5/1/black!20,.5/4/2/white,4/4.5/3/black!20,4.5/7.5/4/white,7.5/11/5/black!20,11/15/6/white}{\draw[very thick,fill=\d] (\a,2.5) rectangle (\b,3.5) node[pos=.5] {\LARGE$\c$};}
				\foreach \a/\b/\c/\d in {0/.5/1/black!20,.5/5/2/white,5/5.5/3/black!20,5.5/10/4/white,10/10.5/5/black!20,10.5/15/6/white}{\draw[very thick, fill=\d] (\a,1.5) rectangle (\b,0.5) node[pos=.5] {\LARGE$\c$};}
				\foreach \i in {0,5,10,15}{\draw[teal,line width=.6mm,dashed] (\i,0) -- (\i,4);}
		\end{tikzpicture}}
		\caption{Illustration of Grimmett's method~\cite{Grimmett04} for the two elections from \cref{tbl:apportia}. Intervals corresponding to left-wing parties are gray. Dashed lines indicate the position of integers before the random shift.}
		\label{fig:grimmett}
	\end{figure}
	
	Did the left-wing parties have a chance of forming a majority coalition?
	The proportional number of votes needed to win a seat in Apportia is $1100/11 = 100$.
	Thus, after Grimmett's method had awarded each party its lower quota, the left-wing coalition collectively controlled three seats (one from party~1, two from party~3, and none from party~5), the right-wing parties controlled five, with three seats remaining unapportioned.
	Hence, a left-wing majority would have only been possible if all remaining seats had gone to the left-wing parties.
	Upon inspection of the residues, however, one can see that this was impossible: Assuming we order the parties numerically for Grimmett's algorithm (see \cref{fig:grimmett}, top), the intervals for parties~1 and~3 are contained within the same unit interval, which means that it is impossible for both intervals to contain an integer, no matter the random shift $u$.
	Since a left-wing majority would have required each of the three parties to receive another seat, the right-wing parties obtained a majority with probability 1.
	(Arguably, this majority is deserved since the right-wing parties received 65\% of the votes.)
	
	Nonetheless, this poor showing by the left-wing parties was discouraging for the voters of party~5, which is why, for the next election, 60 of these voters switch to vote for right-wing parties: 20 voters switch to party~2, 30 to party~4, and 10 to party~6.
	(The right-most columns in \cref{tbl:apportia} show the new tallies.)
	This shift leaves the lower quotas unchanged, but changes the residues.
	As can be seen on the bottom of \cref{fig:grimmett}, as long as the random shift $u$ falls in the range $(9/10, 1)$, all three left-wing parties will get an extra seat.
	This is a paradox: a shift of voters from left-wing to right-wing parties \emph{increased} the chances of a left-wing government from zero to $1/10$.
	In other words, even though the vote shares of all left-wing parties weakly dropped and the vote shares of all right-wing parties grew, the left-wing coalition might well profit from (and hence encourage) this change.
	Our monotonicity axioms will rule out such paradoxes.\footnote{It may appear as if this example relied on the ordering of the parties, so a natural attempt to fix Grimmett's algorithm would be to randomize the order. In fact, \citet{Grimmett04} explicitly proposes randomizing the order as ``it seems desirable to reduce to a
		minimum any correlations which depend on this extraneous element.'' However, the reader may verify that there was no order in which the left-wing parties could have reached a majority in the previous election. Thus, overcoming this paradox will require new apportionment methods.} 
	
	\subsection{Our Approach and Results}
	Our first contribution is the introduction of a family of monotonicity axioms that can guide the choice and future development of randomized apportionment methods.
	Unlike the work of \citet{GPP22}, who lift classical monotonicity axioms into the randomized realm through a rather technical notion of decomposition (see \cref{sec:relatedwork}), our axioms are \emph{native to randomized apportionment} and their violation directly implies paradoxes and perverse incentives.
	Conceptually, our axioms express that a set of parties that gains in votes while other parties lose votes should be entitled to higher chances of \emph{simultaneous representation}, not just to higher chances of representation for each party separately, which is implied by ex-ante proportionality.
	Our axioms differ in which gains and losses of votes they apply to, how they define simultaneous representation, and whether they consider just a single, growing coalition or the comparison of a growing coalition with a shrinking one.
	
	In \cref{sec:rounding}, we begin our investigation in the domain of rounding rules.
	We find that not only the systematic rounding underlying Grimmett's method, but also pipage rounding and the rounding rule with maximum entropy (i.e., \emph{conditional Poisson sampling}) violate simple forms of our monotonicity, and hence admit the paradox described in the previous subsection.
	
	One rounding rule from the statistics literature, however, bucks this trend: we show that \emph{Sampford sampling}~\cite{Sampford67} satisfies the axiom of \emph{selection monotonicity}.
	\begin{restatable}[Sampford rounding satisfies selection monotonicity]{theorem}{thmsampfordmain}\label{thmSampfordMainResult}
		Let $\vec{p}, \vec{p}' \in [0, 1)^n$ be two residue vectors, summing up to the same integer $k$.
		Let $T$ be a set of $k$ parties such that $p_i' \geq p_i$ for all parties $i \in T$ and $p_i' \leq p_i$ for all $i \notin T$.
		Then, the probability that Sampford sampling rounds up the set $T$ is at least as high for $\vec{p}'$ as for $\vec{p}$.
	\end{restatable}
	Proving this theorem reduces to establishing a clean inequality between the partial derivatives of Sampford's probability mass function with respect to changes in the residues.
	But establishing this inequality is a considerable technical challenge because the Sampford's probability-mass expression contains a denominator that sums over all sets of $k$ parties.
	The key to our main result is establishing a surprising equality between this denominator and an expected-value term related to a Poisson trial, which allows us to bound and take derivatives of the denominator.
	Sampford sampling is also the only rounding rule among those we consider that satisfies a two-coalition version of selection monotonicity, which is logically incomparable to the axiom above. 
	Because Sampford sampling possesses the same properties that computer scientists have leveraged when applying dependent randomized rounding (in particular, strong concentration properties like \emph{conditional negative association}~\cite{BJ12}), replacing pipage rounding with Sampford sampling adds our monotonicity axiom ``for free'', which we believe has potential applications, for example in mechanism design. Sampford sampling is the first, and only, rounding rule known to satisfy selection monotonicity, and there are no other obvious contenders. As for general results about this axiom, we are able to show that it implies that the rounding rule be Lipschitz continuous.
	
	In \cref{sec:apportionment}, we return to the general apportionment setting and consider even stronger axioms of monotonicity, which refer to the probability of coalitions exceeding arbitrary seat thresholds.
	We show that, if the axioms are defined with respect to shifts in the raw vote counts (rather than the parties' share of votes), such axioms cannot be satisfied by any apportionment method.
	Even when shifts are measured in terms of vote share, the two-coalition axiom is incompatible with a natural regularity condition on randomized apportionment methods, namely, having \emph{full support}.
	Whether the most direct generalization of our axiom, \emph{threshold monotonicity}, can be satisfied, is left as an open question, though we conjecture that it is indeed satisfied by Sampford sampling.
	For the restricted case where coalitions consist of only two parties, we show that threshold monotonicity is satisfied by Grimmett's method.

	\subsection{Related Work}
	\label{sec:relatedwork}
	Deterministic apportionment has been the subject of a deep axiomatic treatment, much of it developed by Balinski and Young~\cite[see][]{BY01,RU10}.
	Two of the most prominent axioms in this literature, like the ones we develop here, are forms of monotonicity that rule out paradoxes~\cite[Ch.\ 5]{BY01}: population monotonicity, which we explained already, and \emph{house monotonicity}, which states that increasing the house size may not cause any party to lose seats.
	Under mild assumptions, population monotonicity is incompatible with quota~\cite[Thm.\ 6.1]{BY01} and is satisfied exactly by the widely used family of divisor methods~\cite[Thm.\ 4.2]{BY01}.
	Coalitions have mainly been considered in this literature in the context of whether apportionment methods encourage multiple parties to merge into a single party~\cite{BY78}; in randomized apportionment, merging leaves the coalition's expected number of seats unaffected.
	In addition, population monotonicity can be restated in terms of coalitions, in a way that resembles our pairwise axioms: if all parties in coalition $T_1$ weakly gain in votes and all parties in $T_2$ weakly lose votes, then each party in $T_1$ must weakly gain seats or each one in $T_2$ must weakly lose seats.

	We have already touched on the literature on \pips{} sampling without replacement above.
	A central motivation for this literature is the \emph{Horvitz--Thompson (HT) estimator}~\cite{HT52}, which we illustrate in a toy example (loosely following \cite{Sampford67}):
	Suppose we have farms numbered $1, \dots, n$, and each farm $i$ has a known size of land $x_i$, and an unknown yield $y_i$.
	Following \citet{HT52}, we can get an unbiased estimate of the total yield, $\sum_{1 \leq i \leq n} y_i$, by sampling $k$ farms with probability proportional to their size $x_i$ and without replacement; summing up the yield per acre $y_i/x_i$ for each sampled farm $i$, and multiplying the result by $\sum_{1 \leq i \leq n} x_i / k$.\footnote{A simpler estimator would sample $k$ farms uniformly at random, add up their yields $y_i$, and multiply the sum by $n/k$. This estimator is also unbiased but can have much higher variance because it often misses large farms with high yield.}
	The variance of this estimate, which depends on how homogeneous the yield per acre is across farms, can be estimated from the same sample~\cite{YG53} as long as the joint sampling probabilities of pairs of farms are known.
	As a result, the \pips{} literature scrupulously studied pairwise inclusion probabilities~\cite[e.g.,][]{BH83}, but paid little attention to higher-order correlations, which gives our work a distinct lens on the topic.
	
	By contrast, higher-order correlations have been very central to the study of dependent randomized rounding in computer science.
	As we explain in \cref{sec:dependentroundingimplications}, typical applications of dependent rounding lean heavily on concentration bounds.
	In a pursuit of stronger concentration results, researchers on dependent rounding have leveraged increasingly strong notions of negative correlation, ranging from the relatively weak \emph{negative cylinder dependence}~\cite{PS97} %
	over \emph{negative association}~\cite{Wajc17} to the \emph{strong Rayleigh property}~\cite{BBL09}.
	The axioms we pursue fundamentally differ from those above in that they are \emph{relational} rather than \emph{punctual}~\cite{Thomson11}, i.e., in that they constrain how the rounding rule's distributions relates across different residue vectors, rather than just imposing separate constraints for each input.
	Whereas we only consider dependent rounding subject to a cardinality constraint~\cite{Srinivasan01,GKP+06}, dependent rounding schemes have since been developed for more general combinatorial constraints, including bipartite matchings~\cite{GKP+06} and matroids~\cite{CVZ10,PSV17}.
	
	A source of inspiration for our paper is the work by \citet{GPP22}, which, like us, aims to refine the space of randomized apportionment methods by imposing monotonicity axioms in addition to quota and ex-ante proportionality.
	The monotonicity axioms they arrive at, however, have a different flavor from ours:
	They say that a randomized apportionment method satisfies population monotonicity (respectively, house monotonicity) if it can be represented as a probability distribution over deterministic apportionment methods satisfying population monotonicity (respectively, house monotonicity).
	These axioms can be interpreted as preventing paradoxes, but only in artificial scenarios in which the votes or house size change \emph{after the randomization of the apportionment has been performed}.
	The authors show that population monotonicity is incompatible with quota, but that house monotonicity can be guaranteed alongside quota and ex-ante proportionality.
	In \cref{app:housewins}, we show that house monotonicity is unlikely to be compatible with our axioms: no apportionment method satisfying full support, and none of the methods we consider in this paper, satisfy house monotonicity. 
	
	\section{Preliminaries}
	\label{sec:prelims}
	Throughout this paper, let $[n] \coloneqq \{1, \dots, n\}$ denote a set of \emph{parties}.
	For a set $S$ and natural number $k$, we set $\binom{S}{k} \coloneqq \{T \subseteq S \suchthat \abs{T} = k\}$. An \emph{apportionment method} $\app$ is a function that takes in the number of \emph{votes} $v_1, \dots, v_n \in \mathbb{R}_{\geq 0}$ for each party and house size $h$, and returns a vector-valued random variable $\app(\vec{v}, h) \in \mathbb{N}^n$, summing to $h$, which lists the number of seats awarded to each party.
	For a fixed vector of votes $\vec{v}$ and a house size $h$, we denote party $i$'s \emph{standard quota} (i.e., its proportionally deserved share of seats) by $q_i \coloneqq h \, v_i/\sum_{j \in [n]} v_j$ and its \emph{residue} by $p_i \coloneqq q_i - \lfloor q_i \rfloor$.
	We will only consider apportionment methods that satisfy the following two axioms of \citet{Grimmett04}:
	\begin{description}
		\item[Quota.] For any $\vec{v}, h$, the number of seats awarded to party $i$ is always $\lfloor q_i \rfloor$ or $\lceil q_i \rceil$ ex post.
		\item[Ex-ante proportionality.] For any $\vec{v}, h$, the expected number of seats awarded to party $i$ equals $q_i$.
	\end{description}

	Similarly, a \emph{rounding rule} $r$ is a function that takes in a vector of \emph{target probabilities} $\Vec{p}\in [0,1)^{n}$ for each party, which must sum up to an integer $k$, and returns a set-valued random variable $r(\vec{p}) \in \binom{[n]}{k}$, i.e., the values taken on by the random variable are sets of $k$ parties. In addition, $r$ must satisfy for each $\vec{p}$ and $i \in [n]$ that $i$'s probability of being included in $r(\vec{p})$ equals its target probability $p_i$. We are overloading $p_i$ by using it both for general target probabilities and residues induced by a vote profile $\Vec{v}$. We do so as they are mathematically equivalent objects and will play the same role in \Cref{sec:rounding} on rounding rules and \Cref{sec:apportionment} on apportionment rules, respectively. We say that a rounding rule is \emph{continuous} if it is a continuous as a function
	from its domain, thought of as a subset of $\rr^n$, to $\rr^{\binom{n}{k}}$.

	We can easily translate apportionment methods into rounding rules and vice versa.
	On the one hand, we can interpret any (quota and ex-ante proportional) apportionment method $\app$ as a rounding rule: for given target probabilities $\vec{p}$ summing to $k \in \mathbb{N}$, $\app(\vec{p}, k)$ is a random variable ranging over integer vectors adding to $k$. Quota ensures that this vector is binary and can thus be interpreted as a random set of $k$ parties. Ex-ante proportionality ensures that the target probabilities are met. Conversely, any rounding rule $r$ \emph{induces an apportionment method}, as we have described in the introduction:
	Given the votes $\vec{v}$ and house size $h$, one computes the standard quotas $q_i$ and residues $p_i$, awards each party its lower quota $\lfloor q_i \rfloor$, at which point some number $k = \sum_{i \in [n]} p_i$ of seats remain unallocated.
	By randomly awarding an extra seat to the $k$ parties in $r(\vec{p})$, one obtains an apportionment method satisfying quota and ex-ante proportionality.

	We define axioms for rounding rules and apportionment methods directly in \cref{sec:rounding,sec:apportionment}.
	
	\section{Monotonicity of Rounding Rules}
	\label{sec:rounding}
	We begin our investigation of montonicity in the subsetting of rounding rules, for two reasons:
	First, the rounding setting allows us to capture the paradox in our motivating example (\cref{sec:motivating}) in a more direct, circumscribed way: it is paradoxical that an exodus of voters from a bloc of parties could increase the probability that \emph{exactly these parties receive an extra seat}.
	The rounding setting allows us to analyze this paradox\,---\,and bring it to an end\,---\,before considering in \cref{sec:apportionment} more expansive definitions that extend to situations in which the lower quotas, and thus the meaning of getting an ``extra seat,'' change.
	
	The second reason for starting in the rounding setting is the importance of this setting beyond apportionment.
	In particular, our positive results in this section immediately contribute to the understanding of \pips{} sampling, and have potential implications for dependent randomized rounding in computer science, as we sketch in \cref{sec:dependentroundingimplications}.
	
	\subsection{Axioms and Rounding Rules}
	\label{sec:rounding:defs}
	We rule out the paradox from our introduction with the axiom of \emph{selection monotonicity}, which says that, if the target probabilities of all parties in some set $T$ go up, and the target probabilities of all other parties go down, it should become (weakly) more likely that the set $T$ is rounded up:
	\begin{definition}[Selection monotonicity]
		Let $\vec{p}, \vec{p}' \in [0, 1)^n$ be two residue vectors, summing up to the same integer $k$.
		Let $T$ be a set of $k$ parties such that $p_i' \geq p_i$ for all parties $i \in T$ and $p_i' \leq p_i$ for all $i \notin T$.
		A rounding rule $r$ satisfies selection monotonicity if it always holds that
		\[ \mathbb{P}_{S \sim r(\vec{p}')}[S = T] \geq \mathbb{P}_{S \sim r(\vec{p})}[S = T]. \]
	\end{definition}
	
	As a warm-up to our proof, we will also consider a variant which, like the axiom of population monotonicity in classical apportionment theory, implies a disjunction: if a coalition $T_1$ is growing and a coalition $T_2$ is shrinking, $T_1$'s probability of being jointly rounded up should increase \emph{or} $T_2$'s probability of being jointly rounded up should decrease:
	\begin{definition}[Pairwise selection monotonicity]
		Let $\vec{p}, \vec{p}' \in [0, 1)^n$ be two residue vectors, summing up to the same integer $k$.
		Let $T_1, T_2$ be two sets of $k$ parties (which might overlap) such that $p_i' \geq p_i$ for all parties $i \in T_1$ and $p_i' \leq p_i$ for all $i \in T_2$.
		A rounding rule $r$ satisfies pairwise selection monotonicity if it always holds that
		\[ \mathbb{P}_{S \sim r(\vec{p}')}[S = T_1] \geq \mathbb{P}_{S \sim r(\vec{p})}[S = T_1] \quad \text{or} \quad \mathbb{P}_{S \sim r(\vec{p}')}[S = T_2] \leq \mathbb{P}_{S \sim r(\vec{p})}[S = T_2]. \]
	\end{definition}
	Note that pairwise selection monotonicity is not implied by selection monotonicity, since the target probabilities of parties outside of $T_1$ and $T_2$ are allowed to vary arbitrarily.
	
	We now introduce several rounding rules from the literature on \pips{} sampling:
	\begin{description}
		\item[Systematic rounding~\cite{Madow49}.] We have already introduced this rule in the introduction, in the context of Grimmett's randomized apportionment method~\cite{Grimmett04}.
		A random shift $u \sim \mathit{Uniform}([0,1))$ is drawn and party $i$ is selected if the interval $[u + \sum_{1 \leq j < i} p_i, u + \sum_{1 \leq j \leq i} p_i)$ contains an integer.
		\item[Pipage rounding~\cite{DT98,Srinivasan01}.\footnotemark]\footnotetext{This method was first defined in statistics as the \emph{pivotal method}~\cite{DT98}, but we use the term common in computer science.}
		While at least two parties $i, j$ with target probabilities $0 < p_i, p_j < 1$ exist, fix the first two such parties.
		If $p_i + p_j \leq 1$, update $p_i \gets p_i + p_j, p_j \gets 0$ with probability $p_i/(p_i + p_j)$, and update $p_i \gets 0, p_j \gets p_i + p_j$ with the remaining probability.
		Else, update $p_i \gets 1$, $p_j \gets p_i + p_j - 1$ with probability $(1-p_j)/({2 - p_i - p_j})$, and update $p_i \gets p_i + p_j - 1, p_j \gets 1$ with the remaining probability.
		Once no such $p_i, p_j$ exist, select the $k$ parties with target probability $1$.
		\item[Conditional Poisson rounding~\cite{CDL94}.]
		This rounding rule implements the probability distribution over $\binom{[n]}{k}$ with maximal entropy, subject to attaining the target probabilities.
		To implement this rule algorithmically, one must determine the unique \emph{working probabilities} $\pi_1, \dots, \pi_n \in \mathbb{R}_{\geq 0}$ such that choosing each set $T \in \binom{[n]}{k}$ with probability $\prod_{i \in T} \pi_i$ satisfies the target probabilities, which can be achieved using Newton's method~\cite{Tille06}.
		\item[Sampford rounding~\cite{Sampford67}.]
		This rounding rule can be implemented with the following rejective procedure:
		Randomly sample a party such that each party $i$ is selected with probability proportional to $p_i$; call the sampled party $i_1$. Then, randomly sample $k-1$ parties $i_2, \dots, i_k$, with replacement from $[n]$, choosing each party $i$ with probability proportional to $\frac{p_i}{1-p_i}$.
		If the $k$ drawn parties are distinct, select them; else, start over.
		Sampford's rounding rule can also be implemented in a non-rejective, polynomial-time manner~\cite{Grafstrom09}.
	\end{description}
	
	It is worth dwelling on how \emph{strange} (or, in the words of \citet{Tille06}, ``very ingenious'') Sampford's rounding rule is.
	If, say, all $k$ units were drawn with probability proportional to $\frac{p_i}{1 - p_i}$, each set $A$ of $k$ parties would be drawn with probability proportional to $\prod_{i \in A} \frac{p_i}{1 - p_i}$, which means that this would just be Conditional Poisson rounding, but for incorrect working probabilities $\pi_i = \frac{p_i}{1 - p_i}$.
	But, whereas the computation of the correct working probabilities for conditional Poisson rounding has no closed form, adjusting the selection probabilities for only one of the $k$ draws perfectly attains the desired inclusion probabilities (which is nontrivial to show).
	
	Though it is not a rounding rule in our model, our analysis will frequently refer to \emph{Poisson sampling} (not to be confused with \emph{conditional} Poisson sampling).
	Poisson sampling is not a rounding rule because it samples sets of nondeterministic size:
	it simply performs, for each party $i$, an independent Bernoulli trial with success probability $p_i$, and contains $i$ in its output if this trial is a success.
	Clearly, the probability of sampling a set $A$ is just $\prod_{i \in A} p_i \prod_{i \notin A} (1 - p_i)$.
	
	\subsection{Rules that Violate Selection Monotonicity}
	\label{sec:rounding:negative}
	The axiom of selection monotonicity might look so innocuous that one would take its satisfaction for granted. For example, in Poisson sampling, an increase in the target probabilities of $T$ and a decrease in the other parties' target probabilities would obviously increase $\prod_{i \in T} p_i \prod_{i \notin T} (1 - p_i)$, the probability of sampling $T$. Alas, Poisson sampling need not select the right number of parties.
	Selection monotonicity would also be straightforward to show if we used the na\"ive working probabilities $\pi_i = {p_i}/({1-p_i})$ in conditional Poisson rounding.
	However, even though the selection probabilities asymptotically converge to the targets~\cite[Thm.\ 7.3]{Hajek81}, they are not exactly met.\footnote{These two processes are closely related: the Poisson sample, conditioned on selecting exactly $k$ parties, has the same distribution as this na\"ive conditional Poisson rounding\,---\,hence the name.}
	
	Even though these ``almost rounding rules'' easily satisfy selection monotonicity, most proper rounding rules fail the axiom.
	Indeed, our example from the introduction already proved that systematic rounding (even with a randomized ordering of parties) fails selection monotonicity, and we give a similar counterexample for pipage rounding in Appendix \ref{appCounterexamples}.
	\begin{restatable}{proposition}{proPipageCounter}\label{proPipageCounter}
		Systematic rounding and pipage rounding violate selection monotonicity, even if the order of parties is uniformly shuffled.
	\end{restatable}
	
	But, perhaps, this violation is to be expected in light of several characteristics common to systematic and pipage rounding.
	Both rounding rules depend a lot on the ordering of parties, can have strange correlations between the inclusion of pairs of parties, and tend to have probability distributions supported on only few $k$-subsets of parties.
	Sharing none of these properties, and being the ``most natural'' rounding rule given its maximum entropy characterization, conditional Poisson rounding would seem like a promising candidate.
	
	Still, to our own surprise, conditional Poisson rounding also violates selection monotonicity. How we found the counterexample attests to how close conditional Poisson rounding gets to satisfying the axiom: For the special case of $n=6$, $k=3$, we derived a polynomial over $p_1, \dots, p_6$ which is negative exactly for the vectors of target probabilities that would certify a violation of selection monotonicity. Out of 705 monomials, only a single one had a negative sign, which by setting the target probabilities just right, at different orders of magnitude, could be made to overpower all other monomials and render the term as a whole negative. We refer to \Cref{appCounterexamples} for the proof.
	\begin{restatable}{proposition}{proPoissonCounter}\label{proPoissonCounter}
		Conditional Poisson rounding violates selection monotonicity.
	\end{restatable}

	A final reason to temper optimism about satisfying selection monotonicity all too easily is that a strengthened axiom, in which the target probabilities outside $T$ may increase or decrease, cannot be satisfied by any rounding rule.
	To see this, consider a scenario with four parties, which initially all have residues $1/2$ (thus, $k=2$ parties will be selected).
	By symmetry, we may assume that the pair $T \coloneqq \{1,2\}$ is selected with probability at least $1/\binom{4}{2} = 1/6$.
	Now let party~3's target increase to 99\% and party~4's target decrease to 1\%.
	Since party~3 must now almost always be one of the two selected parties, the pair $\{1,2\}$ must be jointly selected much more rarely than $1/6$, even though their target probabilities weakly increased.
	This example shows that the distribution of target probabilities outside of $T$ can necessitate major changes to the joint selection probability of $T$.
	
	\subsection{Sampford Rounding is Monotone}
	\label{sec:rounding:positive}
	We now more closely consider Sampford rounding, the rounding rule that finally satisfies our selection monotonicity axioms.
	Let $S \coloneqq r_{\text{Sampford}}(\vec{p})$ be a random variable describing the set of parties drawn by Sampford rounding.
	From the rejective sampling procedure described in \cref{sec:rounding:defs}, it follows that the probability of a set $A$ being selected is
	\begin{align}
		\mathbb{P}[S=A] &=\frac{\sum_{i \in A} p_i \prod_{j \in A \setminus i} \frac{p_j}{1 - p_j}}{\sum_{A' \in \binom{[n]}{k}}\sum_{i \in A'} p_i \prod_{j \in A' \setminus i} \frac{p_j}{1 - p_j}}.
		\label{eq:sampfordpmfodds}
	\end{align}
	
	As an easy warm-up, and an indication that Sampford rounding is indeed well suited for monotonicity, we first show that it satisfies the pairwise variant of selection monotonicity:
	\begin{theorem}
		Sampford rounding satisfies pairwise selection monotonicity.
	\end{theorem}
	\begin{proof}
		Fix the target probabilities $\vec{p}, \vec{p}'$, the sample size $k$, and the sets $T_1, T_2$.
		Considering the numerator of the probability mass term in \cref{eq:sampfordpmfodds}, it is easy to see that the numerator of $\mathbb{P}[S=T_1]$ weakly increases when going from $\vec{p}$ to $\vec{p}'$, whereas the numerator of $\mathbb{P}[S=T_2]$ weakly decreases.
		
		Now, we distinguish two cases, depending on whether the denominator, which is common to $T_1$ and $T_2$, increases from $\vec{p}$ to $\vec{p}'$ or not.
		If it increases, then, since the numerator for $T_2$ decreases, the probability of sampling $T_2$ must decrease, and we are done.
		Else, since the numerator for $T_1$ increases, the probability of sampling $T_1$ increases, concluding the proof.
	\end{proof}
	Whereas Sampford rounding satisfies pairwise selection monotonicity quite handily, we show in \Cref{appCounterexamples} that the three other rounding rules all violate this axiom.
	
	What made this proof so easy was that we barely had to engage with the denominator term, a sum over exponentially many terms that indicates the probability of non-rejection for the sampling procedure.
	Engaging with how the probability of non-rejection changes with the target probabilities will be the major obstacle in proving our main result, that Sampford satisfies selection monotonicity.
	For this, we slightly reformulate the probability-mass term, and continue \cref{eq:sampfordpmfodds}:
	\begin{align}
		\mathbb{P}[S=A] &= \frac{\sum_{i \in A} (1-p_i) \prod_{j \in A} \frac{p_j}{1 - p_j}}{\sum_{A' \in \binom{[n]}{k}}\sum_{i \in A'} (1- p_i) \prod_{j \in A'} \frac{p_j}{1 - p_j}}\notag\\ 
		&=\frac{\sum_{i \in A} (1-p_i) \prod_{j \in A} p_j \prod_{j \notin A} (1-p_j)}{\sum_{A' \in \binom{[n]}{k}}\sum_{i \in A'} (1- p_i) \prod_{j \in A'} p_j \prod_{j \notin A'} (1 - p_j)} = \frac{f(A)}{\sum_{A' \in \binom{[n]}{k}} f(A')},\label{eq:sampfordpmffa}
	\end{align}
	where the second equality expands the fraction with $\prod_{i \in [n]} (1-p_i)$, and where we introduce the shorthand $f(A)$ for the term $\sum_{i \in A} (1 - p_i) \prod_{j \in A}  p_j \prod_{j \notin A} (1 - p_j)$, which will appear extensively throughout this proof.

	In the remainder of this section, we will prove \cref{thmSampfordMainResult}, i.e., that Sampford rounding is selection monotone.
	By symmetry, we will assume that $A = [k]$.
	A simple argument (\Cref{app:Sufficient}) shows that it suffices to consider the special case where the target probability of one party in $A$ (w.l.o.g., party $1$) infinitesimally grows and the target probability of one party outside of $A$ (w.l.o.g., party $n$) shrinks by the same infinitesimal amount.
	Since Sampford's probability mass function is differentiable, this reduces to the following inequality of its partial derivatives: 
	\begin{proposition}
		\label{prop:pdevs}
		Let $\Vec{p} \in [0, 1)^n$ such that $\sum_{i \in [n]} p_i = k$ for some integer $1 \leq k \leq n-1$, $p_1 > 0$, and $S= r_{\text{Sampford}}(\Vec{p})$.
		Then,
		\[\left(\pdif{p_1} - \pdif{p_n}\right) \, \mathbb{P}[S = [k]] \geq 0.\]
	\end{proposition}

	A crucial quantity throughout our proof will be $s \coloneqq \sum_{i \in [k]} (1 - p_i)$, which measures %
	how far the residues in $[k]$ fall short of being 1. Alternatively, $s = \sum_{i \in [k]} (1 - p_i) = k - \sum_{i \in [k]} p_i = \sum_{i \notin [k]} p_i$ can be seen as how far the residues outside of $[k]$ exceed 0.

	Clearly, to bound the partial derivative from \Cref{prop:pdevs}, we will eventually need to apply the quotient rule, which requires us to understand and bound the numerator and denominator in \cref{eq:sampfordpmffa}, as well as their derivatives. While the numerator's expression is fairly simple, the core challenge lies in bounding the denominator (and its derivative), which sums up over exponentially many sets $A'$ of parties.
	
	We will overcome this hurdle in \cref{lem:magic}, which establishes a surprising equality between this denominator term and the expected value of a natural variable defined over a Poisson trial. Then, we dedicate \Cref{lem:expdif,lem:expbound,lem:probbound} to finding a concise expression of the derivative of the denominator and clean bounds in terms of our parameter $s$ for both of these terms. This established the foundation for proving \Cref{prop:pdevs}.

	For the remainder of the section (including all lemma statements) we fix $\Vec{p} \in [0,1)^n$ as defined in \Cref{prop:pdevs}. To prove our key \Cref{lem:magic}, we require the following auxiliary lemma:

	\begin{lemma}
		\label{lem:magicsub}
		For all $0 \leq \ell \leq n-1$ it holds that 
		\[ \sum_{A \in \binom{[n]}{\ell+1}} f(A) - \sum_{A \in \binom{[n]}{\ell}} f(A) = (k - \ell)\sum_{A \in \binom{[n]}{\ell}} \prod_{j \in A} p_j \prod_{j \notin A} (1 - p_j).\]
	\end{lemma}
	\begin{proof}
		We start by reformulating the summation as follows:
		\begin{align*}
			\sum_{A \in \binom{[n]}{\ell + 1}} f(A) &= \sum_{A \in \binom{[n]}{\ell +1}} \sum_{i \in A} (1 - p_i) \, \prod_{j \in A} p_j \prod_{j \notin A} (1 - p_j)\\
			&= \sum_{A \in \binom{[n]}{\ell +1}} \sum_{i \in A} p_i \, \prod_{j \in A \setminus \{i\}} p_j \prod_{j \notin A \setminus \{i\}} (1 - p_j). \\
			\intertext{In the next step, we will relabel these sums over $A$ and $i$. 
				Specifically, we apply a bijection from the set $\{(A, i) \mid A \in \binom{[n]}{\ell+1}, i \in A\}$ to the set $\{(A', i') \mid A' \in \binom{[n]}{\ell}, i' \notin A'\}$, which maps $(A, i)$ to $(A \setminus \{i\}, i)$. 
				To see that this is a bijection, observe that it has the inverse $(A', i') \mapsto (A' \cup i', i')$. 
				Then, the previous summation is equal to }
			&= \sum_{A \in \binom{[n]}{\ell}} \sum_{i \notin A} p_i \, \prod_{j \in A} p_j \prod_{j \notin A} (1 - p_j), \\
			\intertext{where we omit primes for the sake of readability. 
				Finally, we make use that $\sum_{i \in [n]} p_i = k$ to get}
			&= \sum_{A \in \binom{[n]}{\ell}} \left(k - \ell + \ell - \sum_{i \in A} p_i\right)\, \prod_{j \in A} p_j \prod_{j \notin A} (1 - p_j) \\
			&= \sum_{A \in \binom{[n]}{\ell}} (k - \ell) \,  \prod_{j \in A} p_j \prod_{j \notin A} (1 - p_j) + \sum_{A \in \binom{[n]}{\ell}} \left(\sum_{i \in A} 1 - p_i\right) \,  \prod_{j \in A} p_j \prod_{j \notin A} (1 - p_j) \\
			&= (k - \ell) \,\sum_{A \in \binom{[n]}{\ell}}  \prod_{j \in A} p_j \prod_{j \notin A} (1 - p_j) + \sum_{A \in \binom{[n]}{\ell}} f(A), 
		\end{align*}
		which concludes the proof. \qedhere
	\end{proof}

	\Cref{lem:magicsub} provides us with the necessary tool to prove the reformulation of the denominator. In the new expression, the right-hand side is reminiscent of the variance of the random variable expressing the size of a set sampled via a Poisson trial (a \emph{Poisson binomial variable}). However, the expectation is taken over the absolute distance, rather than the squared distance, towards the mean. 
	
	\newcommand{\expbone}{\mathbb{E}\big[\bone\{|B| < k\} \!\cdot\! (k - |B|) \big]}
	\newcommand{\expabs}{\mathbb{E}\big[\big||B| - k\big|\big]}
	\newcommand{\probterm}{\mathbb{P}\big[|\hat{B}| = k - 1\big]}
	\begin{lemma}
		\label{lem:magic}
		Let $B \subseteq [n]$ be a random variable distributed according to a Poisson trial, i.e., $B$ contains each $1 \leq i \leq n$ independently with probability $p_i$.
		Then,
		\[ \sum_{A \in \binom{[n]}{k}} f(A) = \expbone = \tfrac{1}{2} \cdot \expabs.\]
	\end{lemma}
	\begin{proof}
		For the first equality, we simply write its left-hand side as a telescoping sum:
		\begin{align*}
			\sum_{A \in \binom{[n]}{k}} f(A) &= \sum_{A \in \binom{[n]}{k}} f(A) - {\sum_{A \in \binom{[n]}{0}} f(A)} \tag{second term is $0$}\\
			&= \sum_{0 \leq \ell \leq k-1} \Big(\textstyle \sum_{A \in \binom{[n]}{\ell+1}} f(A) - \sum_{A \in \binom{[n]}{\ell}} f(A)\Big) \\
			&= \sum_{0 \leq \ell \leq k-1} (k - \ell) \cdot \sum_{A \in \binom{[n]}{\ell}} \prod_{j \in A} p_j \prod_{j \notin A} (1 - p_j) \tag{by \cref{lem:magicsub}} \\
			&= \sum_{0 \leq \ell \leq k-1} (k - \ell) \cdot \mathbb{P}\big[|B| = \ell\big] = \expbone.
		\end{align*}
		To show the statement's second equality, recall that the first central moment of any random variable is zero. 
		Applying this to the variable $|B|$ together with the fact that $\mathbb{E}\big[|B|\big] = \sum_{i \in [n]} p_i = k$ yields that $\mathbb{E}\big[ |B| - k\big] = 0$.
		Since $\mathbb{E}\big[ |B| - k\big] = \mathbb{E}\big[\bone\big\{|B| - k \geq 0\big\} \cdot \big(|B| - k\big)\big] + \mathbb{E}\big[\bone\big\{|B| - k < 0\big\} \cdot \big(|B| - k\big)\big]$,
		this implies that
		\begin{align*} 
			\expbone &= - \mathbb{E}\big[\bone\big\{|B| - k < 0\big\} \cdot \big(|B| - k\big)\big]\\ 
			&= \mathbb{E}\big[\bone\big\{|B| - k \geq 0\big\} \cdot \big(|B| - k\big)\big].
		\end{align*}
		The desired equality follows by observing that
		\begin{align*}
			\expabs &= \mathbb{E}\big[\bone\big\{|B| - k \geq 0\big\} \cdot \big(|B| - k\big)\big] + \mathbb{E}\big[\bone\big\{|B| - k < 0\big\} \cdot \big(k - |B|\big)\big]\\ 
			&= 2 \, \expbone, 
		\end{align*}
		which concludes the proof. 
	\end{proof}
	
	Below, we provide a concise upper bound for the denominator in terms of our parameter $s$, building upon the close connection to the variance of the random variable $|B|$ pointed out above. 
	
	\begin{lemma}
		\label{lem:expbound}
		Let $B$ be as in \cref{lem:magic}. Then, $\expabs \leq 2s.$
	\end{lemma}
	\begin{proof}
		We start by upper bounding the relevant term by the variance of $|B|$. That is, 
		\begin{align*}
			\expabs &\leq \mathbb{E}\big[\big(|B| - k\big)^2\big] = \text{Var}\big(|B|\big),
		\end{align*}
		where the first equality uses that $|B|-k$ is integer-valued and the second equality uses that the expected value of $|B|$ is $k$.
		By bounding our term by the variance of $|B|$, we can exploit that $|B|$ is the sum of $n$ independent Bernoulli variables, where the variance of the Bernoulli variable with bias $p_i$ is $p_i \, (1 - p_i)$. 
		Therefore, 
		\begin{align*}
			&\text{Var}\big(|B|\big)= \sum_{i \in [n]} p_i \, (1 - p_i) \leq \sum_{i \in [k]} (1 - p_i) + \sum_{i \notin [k]} p_i = 2 s,
		\end{align*}
		where the last equality follows from $s = \sum_{i \in [k]} (1-p_i) = \sum_{i \notin [k]} p_i$.  \qedhere
	\end{proof}
	
	We now express the derivative of the denominator and link it to the probability that a set sampled from a Poisson trial is of size $k-1$. The proof, which we defer to \Cref{app:deferred:rounding}, takes the derivative for all sizes $\ell$ of $B$ that are relevant in the term of the expectation, i.e., $0 \leq \ell \leq k-1$, and shows that all terms besides one term for $\ell=k-1$ cancel out. 
	
	\begin{restatable}{lemma}{lemxpdif}
		\label{lem:expdif}
		Let $B$ be as in \cref{lem:magic}.
		Let $\hat{B} \subseteq \{2, 3, \dots, n-1\}$ be distributed according to a Poisson trial as above, but omitting parties $1$ and $n$.
		Then,
		\[ \left(\pdif{p_1} - \pdif{p_n}\right) \, \expbone = (p_n - p_1) \, \probterm. \]
	\end{restatable}

	The last step before turning towards proving \Cref{prop:pdevs} is to provide a lower bound for the derivative of the denominator, which we just calculated in \Cref{lem:expdif}. The key idea is to lower bound the probability that $|\hat{B}|$ is of size $k-1$ by the probability that $\hat{B}$ equals exactly $\{2,\dots,k\}$, which we can then use to connect this probability to our parameter $s$. At first sight, the bound appears to be rather weak, e.g., when $s = \sum_{i \not\in [k]}{p_i}>\frac{1}{2}$, the bound is dominated by the trivial lower bound of $0$. However, we will only need this bound in the extreme case that $s<p_1$.

	\begin{lemma}
		\label{lem:probbound}
		Let $\hat{B}$ be as in \cref{lem:expdif}.
		Then,
		\[ \probterm \geq \frac{1 - 2 s}{p_1 (1 - p_n)}. \]
	\end{lemma}
	\begin{proof}
		One way in which $|\hat{B}|$ could equal $k-1$ is if $\hat{B}$ is exactly $\{2, \dots, k\}$. Thus,
		\begin{equation*}
			\probterm \geq \prod_{2 \leq j \leq k} p_j \prod_{k+1 \leq j \leq n-1} (1 - p_j) = \frac{\prod_{j \in [k]}p_j \prod_{k+1 \leq j \leq n} (1-p_j)}{p_1 (1 - p_n)} = \frac{\mathbb{P}\big[B = [k]\big]}{p_1 (1 - p_n)}.
		\end{equation*}
		We lower-bound $\mathbb{P}\big[B = [k]\big]$ using a union bound. Specifically, we need to rule out that any of the Bernoulli trials indexed $1 \leq j \leq k$ fails or that any of the Bernoulli trials indexed $k+1 \leq j \leq n$ succeeds. Thus,
		\[ \mathbb{P}\big[B = [k]\big] \geq 1 - \sum_{1 \leq j \leq k} (1 - p_j) - \sum_{k+1 \leq j \leq n} p_j = 1 - 2s,\]
		since $s=\sum_{i \in [k]} (1-p_i) = \sum_{i \not \in [k]} p_i$.
		The claim follows by combining both inequalities.
	\end{proof}

	We are now ready to prove that the partial derivative of $\Pr[S=[k]]$ is nonnegative. We start by expressing the partial derivatives of the numerator of \Cref{eq:sampfordpmffa}, then apply the quotient rule and finally distinguish the two cases $s\geq p_1$ and $s<p_1$.
	The former case will require little additional effort, while the latter requires the application of \Cref{lem:expbound} and \Cref{lem:probbound}.

	\begin{proof}[Proof of \cref{prop:pdevs}.]
		Recall that we want to bound the partial derivatives of the term $f([k])/{\sum_{A' \in \binom{[n]}{k}} f(A')}$, where $f(A) = \sum_{i \in A} (1 - p_i) \prod_{j \in A}  p_j \prod_{j \notin A} (1 - p_j)$.
		We begin by computing partial derivatives for the numerator $f([k])$:
		\begin{align}
			\pdif{p_1} f([k]) ={} &\textstyle \pdif{p_1} \sum_{i \in [k] \setminus 1} (1 - p_i) \cdot \prod_{j \in [k]} p_j \prod_{j \notin [k]} (1 - p_j) \label{eq:p1fk} \\
			&+ \textstyle \pdif{p_1} 1 \cdot \prod_{j \in [k]} p_j \prod_{j \notin [k]} (1 - p_j) - \pdif{p_1} p_1 \cdot \prod_{j \in [k]} p_j \prod_{j \notin [k]} (1 - p_j) \notag \\
			={} &\textstyle\sum_{i \in [k] \setminus 1} (1 - p_i) \cdot \prod_{j \in [k] \setminus 1} p_j \prod_{j \notin [k]} (1 - p_j)\notag \\ 
			&+ \textstyle\prod_{j \in [k] \setminus 1} p_j \prod_{j \notin [k]} (1 - p_j) - 2 \, p_1 \, \prod_{j \in [k] \setminus 1} p_j \prod_{j \notin [k]} (1 - p_j)\notag \\ 
			={} &\textstyle\sum_{i \in [k]} (1 - p_i) \cdot \prod_{j \in [k] \setminus 1} p_j \prod_{j \notin [k]} (1 - p_j) - \prod_{j \in [k]} p_j \prod_{j \notin [k]} (1 - p_j) \notag \\
			={} &\Big( \frac{1}{p_1} - \frac{1}{\textstyle \sum_{i \in [k]} (1 - p_i)}\Big) \, f([k]) = \Big( \frac{1}{p_1} - \frac{1}{s}\Big) \, f([k]).\notag \\
			\pdif{p_n} f([k]) ={} &- \textstyle \sum_{i \in [k]} (1 - p_i) \cdot \prod_{j \in [k]} p_j \prod_{\substack{j \notin [k]\\j \neq n}} (1 - p_j)  = \displaystyle - \frac{f([k])}{1 - p_n}. \label{eq:pnfk}
		\end{align}
		
		We now apply the quotient rule and plug in equalities we have derived so far:
		\begin{align*}
			&\left( \pdif{p_1} - \pdif{p_n}\right) \, \frac{f([k])}{\sum_{A' \in \binom{[n]}{k}} f(A')} = \left( \pdif{p_1} - \pdif{p_n}\right) \, \frac{f([k])}{\expbone} \tag{by \cref{lem:magic}} \\
			={} &\frac{\left(\left( \pdif{p_1} - \pdif{p_n}\right) f([k]) \right) \!\cdot\!  \expbone - f([k]) \!\cdot\! \left( \pdif{p_1} - \pdif{p_n}\right) \expbone}{\expbone^2} \\
			={} &\frac{\textstyle \left(\frac{1}{p_1} - \frac{1}{s} + \frac{1}{1-p_n} \right) \cdot f([k]) \cdot \expbone - f([k]) \cdot (p_n - p_1) \, \probterm}{\expbone^2} \tag{by \cref{eq:p1fk,eq:pnfk,lem:expdif}} \\
			={} &\underbrace{\frac{f([k])}{\expbone^2}}_{\geq 0} \left(\textstyle \left(\frac{1}{p_1}\! - \!\frac{1}{s} \!+\! \frac{1}{1-p_n} \right) \cdot \expbone \!-\! (p_n \!-\! p_1) \, \probterm \right) 
		\end{align*}
		\newcommand{\sterm}{s} %
		Thus, it suffices to show the nonnegativity of the factor
		\begin{equation}\label{eq:factornonneg}
			\left(\frac{1}{p_1} - \frac{1}{\sterm} + \frac{1}{1-p_n} \right) \cdot \expbone - (p_n - p_1) \, \probterm.
		\end{equation}
		We will prove the nonnegativity by distinguishing two cases, depending on whether $s \geq p_1$ or not.
		Surprisingly, the case $s \geq p_1$\,---\,which intuitively seem more common because the residues of all parties outside of $[k]$ outweigh that of the single party $1$\,---\,is straight-forward.
		The hard case of the proof is when $s < p_1$, i.e., when the residues in $[k]$ are almost 1 and the residues outside $[k]$ almost 0. %
		\medskip
		
		\noindent\textbf{Case $\sterm \geq p_1$.}
		Under this assumption, clearly $1/p_1 - 1/\sterm \geq 0$, which allows us to ignore these terms in \cref{eq:factornonneg}.
		Next, we observe that $\expbone \geq \mathbb{P}\big[|B| = k-1\big]$.
		Moreover, since one way for the number of successes $|B|$ in a Poisson trial being equal to $k-1$ is if there are $k-1$ successes among the Bernoulli variables indexed $2, \dots, n-1$ and additionally the Bernoulli variables indexed $1$ and $n$ both fail, it must hold that $\mathbb{P}\big[|B| = k-1\big] \geq (1 - p_1) \, (1 - p_n) \, \probterm$.
		These two inequalities together imply that
		$\expbone/(1-p_n) \geq (1 - p_1) \, \probterm \geq (p_n - p_1) \, \probterm$, which shows the claimed nonnegativity of the term in \cref{eq:factornonneg}.\medskip
		
		\noindent\textbf{Case $\sterm < p_1$.}
		Note that $p_1 > s = \sum_{i \notin [k]} p_i \geq p_n$.
		Thus, the term $(p_n - p_1) \, \probterm$ we subtract in \cref{eq:factornonneg} is nonpositive, i.e., it can only help us in proving nonnegativity.
		If the coefficient $\frac{1}{p_1} - \frac{1}{s} + \frac{1}{1-p_n}$ was nonnegative, \cref{eq:factornonneg} would be trivially nonnegative; hence, we will assume that it is negative instead.
		Knowing the signs of these coefficients allows us to apply the bounds we have previously derived:
		\definecolor{c1}{rgb}{0.0, 0.3, 0.6}
		\definecolor{c2}{rgb}{0.43, 0.62, 0.28}
		{\allowdisplaybreaks\begin{align*}
				&\left(\frac{1}{p_1} - \frac{1}{\sterm} + \frac{1}{1-p_n} \right) \cdot \expbone - (p_n - p_1) \, \probterm \\
				\geq{} &\left(\frac{1}{p_1} - \frac{1}{\sterm} + \frac{1}{1-p_n} \right) \cdot \sterm - (p_n - p_1) \, \frac{1 - 2 \, \sterm}{p_1 \, (1 - p_n)} \tag{by \cref{lem:magic,lem:expbound,lem:probbound}} \\
				={} &\frac{1}{p_1 \, (1 - p_n)} \cdot \left( (1 - p_n) \, s - p_1 \, (1 - p_n) + p_1 \, s - (p_n - p_1) \, (1 - 2 \, s)\right) \\
				={} &\frac{1}{p_1 \, (1 - p_n)} \cdot \big({s - {p_1 \, s}} + p_1 \, p_n - p_n + {p_n \, s}\big) \tag{note $s-sp_1 \geq 0$} \\
				\geq{} &\frac{1}{p_1 \, (1 - p_n)} \cdot \big( p_n \, (p_1 - 1 + s)\big) \geq \frac{p_n}{p_1 \, (1 - p_n)} \cdot \big(p_1 - 1 + (1 - p_1) \big) = 0.
			\end{align*}
			This establishes that, in both cases, $\left(\pdif{p_1} - \pdif{p_n}\right) \, \mathbb{P}[S = [k]]$ is nonnegative as claimed.}
	\end{proof}

	Since we have established (informally in this section and formally in \Cref{app:Sufficient}) that \Cref{prop:pdevs} implies selection monotonicity, the above proof implies our main result, 
	\cref{thmSampfordMainResult}.
	
	\subsection{Potential Implications for Dependent Randomized Rounding}
	\label{sec:dependentroundingimplications}
	The three properties that turned pipage rounding into such a successful tool\,---\,ex-ante proportionality, selecting exactly $k$ units without replacement, and notions of negative correlation~\cite{Srinivasan01}\,---\,are equally satisfied by Sampford rounding~\cite{BJ12}.
	Hence, in all algorithms that use dependent rounding (subject to a cardinality constraint) as a black-box subroutine~\cite[e.g.,][]{BSS18,CJM+20,CL12}, swapping out pipage rounding for Sampford rounding would preserve the algorithm's guarantees, while gaining Sampford's monotonicity properties ``for free''.
	Though we have yet to explore how to leverage this monotonicity in specific applications, we sketch below one path in which Sampford's monotonicity could add incentive guarantees to algorithms using dependent rounding.
	
	Dependent randomized rounding is commonly used in computer science to develop algorithms, in particular approximation algorithms, for combinatorial problems.
	Dependent rounding can be seen as randomly rounding a point $x$ within an integral polytope $P$ to a vertex $X$ of $P$, ensuring that the expected value of $X$ equals $x$. This is a useful subroutine in many algorithms for NP-hard optimization problems: First, locate a point in the space of the polytope, typically achieved by solving a fractional relaxation of the problem.
	Then, this point is rounded in a way that ensures that a linear function (or a continuous extension of a submodular function) satisfies Chernoff-style concentration bounds.
	These guarantees can then be used to bound the objective value with high probability, thus delivering approximation guarantees.
	As \citet{BJ12} showed, Sampford sampling, like pipage rounding, satisfies the \emph{strong Rayleigh} property, which implies extremely general concentration properties. 
	
	In such applications of dependent rounding, selection monotonicity (and similar monotonicity properties that Sampford rounding might possess, see \Cref{conjecture}) might for example allow to construct approximation algorithms that can be implemented as truthful mechanisms. In our case, the polytope $P$ is simply the base polytope of the uniform matroid, which has proved influential in applications including Steiner tree problems~\cite{Srinivasan01}, $k$-median~\cite{CL12}, committee selection~\cite{CJM+20}, and, recently, online algorithms~\cite{NSW23a}.
	
	Concretely, consider an approximation algorithm that consists of the two steps mentioned above: \emph{(1)}~determine a point $x$ inside $P$, then \emph{(2)}~randomly round $x$ to select a vertex $X$ of $P$.
	Let each agent $i$ have a preferred vertex $X_i$ of $P$ with utility $\tau_i \cdot \mathbb{P}[X=X_i]$, which depends on their private type $\tau_i \in \mathbb{R}$.
	These types are inputs to the approximation algorithm, whose objective need not be related to the utilities.
	Suppose that step~(1) of the algorithm is monotone in the reported types, i.e., whenever agent $i$ increases their reported type, this causes $x_j$ to weakly increase for the $k$ many matroid elements $j$ that correspond to the vertex $X_i$ and $x_j$ to weakly decrease for all other $j$.
	Then, the selection monotonicity of Sampford rounding implies an overall monotonicity: if agent $i$ increases their report, $X_i$ will be more likely to be sampled.
	Hence, by charging the agents Myerson payments~\cite{M81a}, we obtain a mechanism that is incentive compatible in expectation (and hence the approximation algorithm receives truthful inputs).
	Should Sampford sampling even satisfy stronger notions of proportionality (as we conjecture in \cref{sec:thresholdconj}), this blueprint extends to richer utilities; for example, agent $i$ could care about a subset $T_i$ of matroid elements, have a monotone nondecreasing function $f_i : \mathbb{N} \to \mathbb{R}$ that expresses diminishing returns or returns to scale, and the agent's utility would be $\tau_i \cdot \mathbb{E}\left[ f_i(|X \cap T_i|) \right]$.

	\subsection{Selection Monotonicity Implies Continuity}\label{subContinuity}
	
	We note that all rules discussed in this paper are continuous: small changes to the input residues do not result in large changes to the probability that any set is selected. As the following result shows, this is a necessary requirement for selection monotonicity.
	
	\begin{theorem}\label{thmContinuity}
		Any selection monotone rounding rule is Lipschitz continuous.
		Specifically, let $\vec{p}, \vec{p}'$ be two residue profiles, and let $T$ be a set of $k$ parties.
		Then,
		\[ \left| \mathbb{P}\big[r(\vec{p}') = T\big] - \mathbb{P}\big[r(\vec{p}) = T\big] \right| \leq \|\vec{p} - \vec{p}'\|_1. \]
	\end{theorem}
	
	As in the proof of Theorem \ref{thmSampfordMainResult}, we first consider the special case when only two residues are changing at a time.
	
	\begin{lemma}
		\label{lemLipschitzGeneral}
		Consider a selection monotone rounding rule $r$, a residue profile $\vec{p}$, and $0 < \delta < \min(p_2, 1-p_1)$.
		Define a residue profile $\vec{p}'$, in which $p_1' \coloneqq p_1 + \delta$, $p_2' \coloneqq p_2 - \delta$, and all other $p_i' \coloneqq p_i$.
		Then, for any set $T$ of $k$ parties, it holds that
		\begin{equation}\label{equLipshitzBound}
			\left| \mathbb{P}\left[r(\vec{p}') = T\right] - \mathbb{P}\left[r(\vec{p}) = T\right] \right| \leq 2 \, \delta.
		\end{equation}
	\end{lemma}
	\begin{proof}
		We may assume w.l.o.g. that $\mathbb{P}\left[r(\vec{p}') \supseteq \{1,2\}\right] \geq \mathbb{P}\left[r(\vec{p}) \supseteq \{1,2\}\right]$; otherwise, simply change the roles of parties $1$ and $2$, and the roles of $\vec{p}, \vec{p}'$.
		We break the proof up into three cases, depending on how $T$ intersects $\{1, 2\}$.\medskip
		
		\noindent\textbf{Case 1: Sets $T$ such that $1 \in T \not\ni 2$.}
		By ex-ante proportionality, $\sum_{T \in \binom{[n]}{k}, 1 \in T} \mathbb{P}[r(\vec{p})=T] = p_1$ and $\sum_{T \in \binom{[n]}{k}, 1 \in T} \mathbb{P}[r(\vec{p}')=T] = p_1 + \delta$.
		For all $T$ that contain $1$ but not $2$, selection monotonicity ensures that $\mathbb{P}[r(\vec{p}')=T] \geq \mathbb{P}[r(\vec{p})=T]$.
		Hence,
		\begin{align*}
			\delta &= (p_1 + \delta) - p_1 \\
			&= \sum_{\substack{T \in \binom{[n]}{k}\\1 \in T}} \mathbb{P}[r(\vec{p}')=T] - \sum_{\substack{T \in \binom{[n]}{k}\\1 \in T}} \mathbb{P}[r(\vec{p})=T] \\
			&= \sum_{\substack{T \in \binom{[n]}{k}\\1 \in T \not\ni 2}} \big( \mathbb{P}[r(\vec{p}')=T] - \mathbb{P}[r(\vec{p})=T] \big) + \underbrace{\sum_{\substack{T \in \binom{[n]}{k}\\T \supseteq \{1,2\}}} \mathbb{P}[r(\vec{p}')=T] - \sum_{\substack{T \in \binom{[n]}{k}\\T \supseteq \{1,2\}}} \mathbb{P}[r(\vec{p})=T]}_{=\mathbb{P}\left[r(\vec{p}') \supseteq \{1,2\}\right] - \mathbb{P}\left[r(\vec{p}) \supseteq \{1,2\}\right] \geq 0} \\
			&\geq \sum_{\substack{T \in \binom{[n]}{k}\\1 \in T \not\ni 2}} \underbrace{\big( \mathbb{P}[r(\vec{p}')=T] - \mathbb{P}[r(\vec{p})=T] \big)}_{\geq 0}.
		\end{align*}
		Hence, for each set $T$ in the sum, $0 \leq \mathbb{P}[r(\vec{p}')=T] - \mathbb{P}[r(\vec{p})=T] \leq \delta$; i.e., Equation (\ref{equLipshitzBound}) holds for all sets $T$ containing 1 but not 2.
		Note that this chain of inequalities also shows that $\mathbb{P}\left[r(\vec{p}') \supseteq \{1,2\}\right] - \mathbb{P}\left[r(\vec{p}) \supseteq \{1,2\}\right] \leq \delta$, which we will use in the second case.\medskip
		
		\noindent\textbf{Case 2: Sets $T$ such that $2 \in T \not\ni 1$.}
		By ex-ante proportionality, $\sum_{T \in \binom{[n]}{k}, 2 \in T} \mathbb{P}[r(\vec{p})=T] = p_2$ and $\sum_{T \in \binom{[n]}{k}, 2 \in T} \mathbb{P}[r(\vec{p}')=T] = p_2 - \delta$.
		For all $T$ that contain $2$ but not $1$, selection monotonicity ensures that $\mathbb{P}[r(\vec{p}')=T] \leq \mathbb{P}[r(\vec{p})=T]$. Hence,
		\begin{align*}
			\delta &= p_2 - (p_2 - \delta) \\
			&= \sum_{\substack{T \in \binom{[n]}{k}\\2 \in T}} \mathbb{P}[r(\vec{p})=T] - \sum_{\substack{T \in \binom{[n]}{k}\\2 \in T}} \mathbb{P}[r(\vec{p}')=T] \\
			&= \sum_{\substack{T \in \binom{[n]}{k}\\2 \in T \not\ni 1}} \big( \mathbb{P}[r(\vec{p})=T] - \mathbb{P}[r(\vec{p}')=T] \big) + \underbrace{\sum_{\substack{T \in \binom{[n]}{k}\\T \supseteq \{1,2\}}} \mathbb{P}[r(\vec{p})=T] - \sum_{\substack{T \in \binom{[n]}{k}\\T \supseteq \{1,2\}}} \mathbb{P}[r(\vec{p}')=T]}_{=\mathbb{P}\left[r(\vec{p}) \supseteq \{1,2\}\right] - \mathbb{P}\left[r(\vec{p}') \supseteq \{1,2\}\right] \geq -\delta},
			\intertext{which implies \cref{equLipshitzBound} for this case:}
			2 \, \delta &\geq \sum_{2 \in T \not\ni 1} \underbrace{\big( \mathbb{P}[r(\vec{p})=T] - \mathbb{P}[r(\vec{p}')=T] \big)}_{\geq 0}.
		\end{align*}
		\medskip
		
		\noindent\textbf{Case 3: Sets $T$ such that $1, 2 \notin T$ or $1, 2 \in T$.}
		Pick a different party whose membership in $T$ is opposite that of 1 and 2, i.e., if $1, 2 \in T$, then choose a party in $T$, and if $1, 2 \in T$, choose a party outside of $T$. Without loss of generality, let this party be party $3$.
		We may also assume without loss of generality that $\delta$ is less than $1 - p_3$ (which is positive since residues are less than 1, by definition).
		Else, i.e., if $\delta$ is larger, we can split up the increase of $p_1$ by $\delta$ and decrease of $p_2$ by $\delta$ into finitely many steps in which these residues successively increase or decrease, by some $\delta' < 1 - p_3$ in each step.
		Our claim for the original $\delta$ then follows by applying the triangle inequality to the bounds for the smaller steps.
		
		Let $\vec{p}''$ be a new residue profile, where $p''_1 \coloneqq p_1 = p'_1 - \delta$, $p''_2 = p_2 + \delta = p'_2$ and $p''_3 = p_3 + \delta = p'_3 + \delta$. Note that this is a valid residue profile by how we chose $\ell$. The key observation is that we already know the bound in Equation~(\ref{equLipshitzBound}) holds going from $\vec{p}$ to $\vec{p}''$ and from $\vec{p}''$ to $\vec{p}'$, since the residues changing involve one party in $T$ and one party outside of $T$. Thus,
		\begin{align*}
			&\left| \mathbb{P}\left[r(\vec{p}') = T\right] - \mathbb{P}\left[r(\vec{p}) = T\right] \right| \\
			={}& \big| \underbrace{\left(\mathbb{P}\left[r(\vec{p}') = T\right] - \mathbb{P}\left[r(\vec{p}'') = T\right]\right)}_{\in [-2\,\delta, 2\, \delta]} + \underbrace{\left(\mathbb{P}\left[r(\vec{p}'') = T\right] - \mathbb{P}\left[r(\vec{p}) = T\right]\right)}_{\in [-2 \, \delta, 2 \, \delta]} \big|.
		\end{align*}
		Na\"ively, this would yield a bound of $4 \, \delta$, but we can sharpen it by observing that selection monotonicity ensures that both of the terms above have opposite signs: If $1, 2 \notin T$ and $3 \in T$, then the first of these two terms is $\leq 0$ and the second $\geq 0$, whereas if $1, 2 \in T$ and $3 \notin T$, then the first term is $\geq 0$ and the second term is $\leq 0$. Either way, the sum of these terms is at most $2\delta$, so Equation~(\ref{equLipshitzBound}) holds.
	\end{proof}
	
	\begin{proof}[Proof of Theorem \ref{thmContinuity}]
		
		Let $\vec{p}, \vec{p}'$ be an arbitrary pair of residue profiles. Define the following sets of parties:
		\begin{align*}
			S_+ &\coloneqq \{j \in [n] \suchthat p'_j > p_j\},\\
			S_- &\coloneqq \{j \in [n] \suchthat p'_j < p_j\}.
		\end{align*}
		Consider the following algorithm, which transforms $\vec{p}$ into $\vec{p}'$:
		
		{
			\setlength{\interspacetitleruled}{0pt}%
			\setlength{\algotitleheightrule}{0pt}%
			\begin{algorithm}[H]
				$\vec{p}^1 \gets \vec{p}$\;
				$i \gets 1$\;
				\While{$\vec{p}^i \neq \vec{p}'$}
				{
					$j_1 \gets$ any element of $S_+$ such that $p'_{j_1} > p^i_{j_1}$\;
					$j_2 \gets$ any element of $S_-$ such that $p'_{j_2} < p^i_{j_2}$\;
					$\delta_i \gets \min\{p'_{j_1} - p^i_{j_1}, p^i_{j_2} - p'_{j_2}\}$\;
					$\vec{p}^{i + 1} \gets$ copy of $\vec{p}^i$, except $p^{i + 1}_{j_1} \coloneqq p^i_{j_1} + \delta_i$ and $p^{i + 1}_{j_2} \coloneqq p^i_{j_2} - \delta_i$\;
					$i \gets i + 1$\;
				}
			\end{algorithm}
		}
		
		Observe that, in each round of this algorithm, we make a new component of $\vec{p}^i$ equal to that of $\vec{p}'$. Thus, the algorithm terminates after at most $n - 1$ iterations, producing a sequence of residue profiles
		$\vec{p} = \vec{p}^1, \vec{p}^2, \vec{p}^3, \dots, \vec{p^\ell} = \vec{p}'$
		for some $\ell \geq n$. It is easy to see that the total sum of the changes $\delta_i$ at each iteration is precisely the sum of differences between $\vec{p}$ and $\vec{p}'$ across $S_+$ (or, alternatively, across $S_-$). Thus,
		\begin{equation}\label{equTV}
			\sum_{i = 1}^{\ell - 1} \delta_i = \sum_{j \in S^+} (p'_j - p_j) \leq \max_{S \subseteq [n]} \sum_{j \in S} (p'_j - p_j) = d_{\text{TV}}(\vec{p}, \vec{p}') = \frac12 \|\vec{p}' - \vec{p}\|_1,
		\end{equation}
		where we have applied the well-known formula for the total variation distance $d_{\text{TV}}$.\footnote{This formula still holds even when probabilities sum to some constant other than 1. Also, we really have equality in the middle transition as well, but that is not necessary to argue for this proof.}
		
		Note that Lemma~\ref{lemLipschitzGeneral} applies to each consecutive pair of residue profiles $\vec{p}^i$ and $\vec{p}^{i + 1}$, since they differ in only two coordinates by $\delta_i$. Hence, we have
		\begin{align*}
			\left| \mathbb{P}\left[r(\vec{p}') = T\right] - \mathbb{P}\left[r(\vec{p}) = T\right] \right| & \leq \sum_{i = 1}^{\ell - 1} \left| \mathbb{P}\left[r(\vec{p}^{i + 1}) = T\right] - \mathbb{P}\left[r(\vec{p}^i) = T\right] \right| \tag{triangle inequality}\\
			&\leq \sum_{i = 1}^{\ell - 1} 2 \delta_i \tag{from Lemma \ref{lemLipschitzGeneral}}\\
			&\leq 2 \cdot \frac{1}{2} \, \|\vec{p} - \vec{p}'\|_1 \tag{from Equation (\ref{equTV})}\\
			&= \|\vec{p} - \vec{p}'\|_1. \qedhere
		\end{align*}
	\end{proof}

	\section{Monotonicity of Apportionment Methods}
	\label{sec:apportionment}
	Returning to our full setting of randomized apportionment, what does the (pairwise) selection monotonicity of Sampford rounding win us?
	Fix an ``old'' and a ``new'' election with respective votes $\vec{v}, \vec{v}'$ and the same house size $h$, as well as a coalition $T$ of parties.
	Suppose that each party $i \in T$ has increased its share of the vote from the old election to the new one, whereas each party $i \notin T$ has lost in vote share.
	Then, if no party's vote share has shifted by so much that its lower quota changed and the sum of residues equals $|T|$, the apportionment method induced by Sampford rounding gives $T$ a higher chance of jointly being rounded up, by \cref{thmSampfordMainResult}.
	As a result, the paradox from \cref{sec:motivating} cannot happen for this method.
	
	When lower quotas change between elections, however, the entire concept of ``$T$ being jointly rounded up'' stops making much sense.
	Returning to the example from the introduction, suppose that, instead of the exodus of voters from the left-wing party~5, 40 voters from the right-wing party~4 had moved to party~5, which among other things would increase party~5's lower quota from 0 to 1.
	Since this shift also decreases the sum of all residues from 3 to 2, no apportionment method satisfying quota can ``round up'' all three left-wing parties.
	But this line of thinking muddles that the increase of party~5's lower quota is good news for the left-wing coalition: party~5 receiving one seat has turned from a possibility into certainty, and party~5 might even receive a second seat!
	
	In our view, the right generalization of the paradox to changing lower quotas is by considering how likely the coalition is to exceed certain \emph{thresholds} of house seats.
	In our example, the relevant threshold was that of a majority of seats; broadly speaking, party~5 increasing its lower quota ought to increase the prospects of the coalition gaining such a majority.
	Besides the simple majority, a coalition might care about its probability of exceeding other seat thresholds, which might turn it into a supermajority, a blocking minority, or a parliamentary group with certain privileges~\cite[e.g.,][rules 33--34]{EuropeanParliament19}.
	Ideally, an apportionment method should have the property that the probability of $T$ exceeding a threshold is monotone, for all thresholds simultaneously.
	
	In this section, we capture a monotonicity in the probability of crossing thresholds in four axioms.
	Two of them will not be achievable for any randomized apportionment method, and one more axiom is at least impossible if the apportionment method satisfies a natural condition.
	The satisfiability of the fourth monotonicity axiom is the most tantalizing open question of our work; we conjecture that it is satisfied by the Sampford apportionment method, and show that Grimmett's method satisfies it for coalitions of size $T = 2$.
	
	\subsection{Axioms and Apportionment Methods}
	\label{sec:apportionment:defs}
	The apportionment methods we study are induced by the rounding rules from the last section, through the construction defined in \cref{sec:prelims}.
	We refer to the apportionment method induced by conditional Poisson rounding as the \emph{conditional Poisson method}, to the method induced by Sampford rounding as the \emph{Sampford method}, and so forth.
	An exception is the method induced by systematic rounding, which we refer to as Grimmett's method~\cite{Grimmett04}.
	Recall that all of these apportionment methods satisfy quota and ex-ante proportionality.
	
	The main way in which we strengthen monotonicity for apportionment is \emph{threshold monotonicity}: 
	\vspace{-.5cm}
	\begin{definition}[Threshold monotonicity]
		Let $\vec{v}, \vec{v}' \in \mathbb{R}_{\geq 0}^n$ be two vote vectors, and let $h \in \mathbb{N}$.
		Let $T$ be a set of parties such that $q_i' \geq q_i$ for all $i \in T$\footnote{Or, equivalently, we can require that $i$'s share of the vote $\frac{v_i}{\sum_{j \in [n]} v_j} = q_i / h$ increases rather than its quota.} and $q_i' \leq q_i$ for all $i \notin T$.
		An apportionment method $\app$ satisfies threshold monotonicity if it always holds, for all thresholds $\theta \in \mathbb{N}$, that
		\[ \mathbb{P}_{\vec{\alpha} \sim \app(\vec{v}', h)}[\textstyle\sum_{i \in T} \alpha_i \geq \theta] \geq \mathbb{P}_{\vec{\alpha} \sim \app(\vec{v}, h)}[\textstyle\sum_{i \in T} \alpha_i \geq \theta]. \]
	\end{definition}
	In other words, the axiom requires the random variable $\sum_{i \in T} \app(\vec{v}', h)_i$, which describes the total number of seats awarded to $T$ for the votes $\vec{v}'$, to first-order stochastically dominate the corresponding random variable for the votes $\vec{v}$.
	We also define a pairwise analogue to the axiom:
	\begin{definition}[Pairwise threshold monotonicity]
		Let $\vec{v}, \vec{v}' \in \mathbb{R}_{\geq 0}^n$ be two vote vectors, and let $h \in \mathbb{N}$.
		Let $T_1, T_2$ be two sets of parties such that $q_i' \geq q_i$ for all $i \in T_1$ and $q_i' \leq q_i$ for all $i \in T_2$.
		An apportionment method $\app$ satisfies pairwise threshold monotonicity if always at least one of the following two statements holds: 
		\begin{align*}
			\forall \theta \in \mathbb{N},\;&\mathbb{P}_{\vec{\alpha} \sim \app(\vec{v}', h)}[\textstyle\sum_{i \in T_1} \alpha_i \geq \theta] \geq \mathbb{P}_{\vec{\alpha} \sim \app(\vec{v}, h)}[\textstyle\sum_{i \in T_1} \alpha_i \geq \theta] \quad \text{or} \\
			\forall \theta \in \mathbb{N},\;&\mathbb{P}_{\vec{\alpha} \sim \app(\vec{v}', h)}[\textstyle\sum_{i \in T_2} \alpha_i \geq \theta] \leq \mathbb{P}_{\vec{\alpha} \sim \app(\vec{v}, h)}[\textstyle\sum_{i \in T_2} \alpha_i \geq \theta].
		\end{align*}
	\end{definition}
	Though these axioms (for apportionment methods) cannot syntactically be compared to our previous monotonicity axioms (for rounding rules), 
	they should be thought of as strengthenings of their selection equivalents.
	This relationship can be made formal using the embedding from \cref{sec:prelims}: Any apportionment method satisfying (pairwise) threshold monotonicity, when translated into a rounding rule, satisfies (pairwise) selection monotonicity, by selecting the threshold so that all parties in the coalition must be selected to meet the threshold.

	We will show that pairwise threshold monotonicity is incompatible with a natural condition on apportionment methods:
	\begin{definition}[Full support]
		An apportionment method $\app$ satisfies full support if, for all vote vectors $\vec{v}$, house sizes $h$, and apportionment vectors $\vec{\alpha}$ that satisfy quota (i.e., $\vec{\alpha} \in \mathbb{N}^n$, $\sum_{i \in [n]} \alpha_i = h$, and $\alpha_i \in \{\lfloor q_i \rfloor, \lceil q_i \rceil\}$ for all $i \in [n]$), $\app(\vec{v}, h)$ will take on the value $\vec{\alpha}$ with positive probability.
	\end{definition}
	The definitions of Sampford rounding immediately implies that the Sampford method has full support. Similarly, the conditional Poisson method has full support, since, if the residue of any party $i$ is non-zero, its working probability $\pi_i$ must also be positive by ex-ante proportionality.

	Finally, we define variants of the two preceding axioms whose precondition argues about changes not to the standard quotas or vote shares, but to changes in the raw vote counts:
	\emph{Vote-count threshold monotonicity} is defined just like threshold monotonicity, except requiring that $v_i' \geq v_i$ for all $i \in T$ and $v_i' \leq v_i$ for all $i \notin T$, instead of the matching assumptions for the $q_i$ and $q_i'$.
	In the same way, we define \emph{pairwise vote-count threshold monotonicity} like its non-vote-count variant, where we now require that $v_i' \geq v_i$ for all $i \in T_1$ and $v_i' \leq v_i$ for all $i \in T_2$.
	Under the mild assumption that an apportionment method satisfies homogeneity, i.e., that scaling the vote vector by a constant doesn't change the apportionment, (pairwise) vote-count threshold monotonicity implies (pairwise) threshold monotonicity.

	\subsection{Conjecture and $|T|=2$ Possibility for Threshold Monotonicity}
	\label{sec:thresholdconj}
	We leave the most pressing question, finding an apportionment method that satisfies threshold monotonicity, unresolved.
	Based on extensive computational search for counterexamples, however, we conjecture that the Sampford method satisfies this axiom:
	\begin{conjecture}\label{conjecture}
		The Sampford method satisfies threshold monotonicity.
	\end{conjecture}
	Proving this conjecture, of course, appears technically very difficult.
	A major complication relative to our proof of selection monotonicity is that, in many of the apportionment outcomes forming the event ``$T$ receives at least $\theta$ seats'', some of the rounded seats go to parties outside of $T$ whose vote share is shrinking, which means that the probability of many of these outcomes will decrease for the Sampford method.
	Thus, a proof for our conjecture will require a way to charge the decreases in probability of some outcomes to increases in probability of other outcomes.
	
	Still, we do have some limited good news:
	Though Grimmett's method does not satisfy threshold monotonicity for coalitions of three and more parties (not even selection monotonicity, as we saw in \cref{sec:motivating}), it satisfies the axiom for coalitions of size 2.

	\begin{theorem}\label{proGrimmett2}
		Grimmett's method satisfies threshold monotonicity for coalitions $T$ of size $\leq 2$.
	\end{theorem}
	
	\begin{proof}[Proof of Proposition \ref{proGrimmett2}]
		We prove threshold monotonicity of the variant of Grimmett's method that uses an arbitrary fixed ordering of the states, clearly implying threshold monotonicity for a randomized order as well. For our analysis, we will modify our presentation of Grimmett's method in two ways.
		First, instead of giving all parties their lower quotas and then lining up their residues as intervals on the number line, we directly line up their standard quotas as intervals on the number line. Then each party gets a number of seats equal to the number of integers their interval contains. This is clearly equivalent: since extending the interval of each party $i$ from length $p_i$ to length $q_i = p_i + \lfloor q_i \rfloor$ adds an integral length $\lfloor q_i\rfloor$, the party's interval will contain exactly $\lfloor q_i \rfloor$ more integers than previously, and the shift of other parties remains the same modulo 1. Second, we arbitrarily shift all the intervals to the right by some amount $u_0$. Again, this has no bearing on the probability distribution of numbers of seats awarded, since all intervals are additionally shifted by a random $u$ drawn uniformly from $[0, 1)$, and the distribution of $u + u_0$ modulo 1 is identical to that of $u$ modulo 1. We call this equivalent implementation of Grimmett's method {$G(u_0)$}.
		
		\begin{figure}[t]
			\centering
			\scalebox{0.6}{
				\begin{tikzpicture}
					\node [align=right,left] at (-.15,3.25) {\LARGE Previous};
					\node [align=right,left] at (-.15,2.75) {\LARGE election};
					\node [align=right,left] at (-.15,1.25) {\LARGE New};
					\node [align=right,left] at (-.15,0.75) {\LARGE election};
					
					\foreach \x/\y/\z/\w in {0/3/A/white,3/6/B/black!20,6/9/C/white,9/11.5/D/black!20,11.5/15/E/white}{\draw[very thick,fill=\w] (\x,2.5) rectangle (\y,3.5) node[pos=.5] {\LARGE$\z$};}
					\foreach \x/\y/\z/\w in {1/2/A'/white,2/7/B'/black!20,7/9/C'/white,9/12.5/D'/black!20,12.5/16/E'/white}{\draw[very thick, fill=\w] (\x,1.5) rectangle (\y,0.5) node[pos=.5] {\LARGE$\z$};}
					\draw[very thick, pattern=north west lines, pattern color=black!20] (0,1.5) rectangle (1,0.5) node[pos=.5] {\LARGE$u_0$};
					\foreach \i in {0,5,10,15}{\draw[teal,line width=.6mm, dashed] (\i,0) -- (\i,4);}
			\end{tikzpicture}}
			\caption{Illustration of the proof of \Cref{proGrimmett2} stating that Grimmett's method satisfies threshold monotonicity for $|T|=2$. Intervals corresponding to growing parties are gray.}
			\label{fig:grimmett2}
		\end{figure}
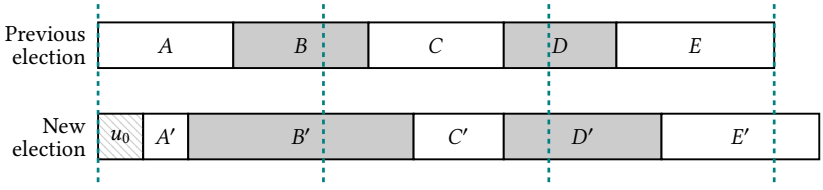
		
		If $\abs{T} = 1$, there is nothing to show: proportionality alone implies threshold monotonicity. So suppose $\abs{T} = \{i, j\}$, where $1 \leq i < j \leq n$. First consider running $G(0)$ on an arbitrary input {$(\Vec{v},h)$}. Before adding the random shift $u$, let $A \cup B \cup C \cup D \cup E$ be the partition of $[0, h)$ where $A$ is the union of all intervals for parties numbered lower than $i$, $B$ is the interval for party $i$, $C$ is the union of all intervals numbered between $i$ and $j$, $D$ is the interval for party $j$, and $E$ is the union of all intervals for parties numbered higher than $j$.
		Though we know that the sizes of the intervals corresponding to $i$ and $j$ increase when going from $(\Vec{v}, h)$ to $(\Vec{v}',h)$ for $G(0)$, it need not be true that the new intervals contain the previous ones. To overcome this issue, we show in \Cref{lemGrimmett2Tech} in \Cref{app:GrimmettTech} that there exists a shift $u_0 \in \mathbb{R}$ such that, when $A'$, $B'$, $C'$, $D'$, and $E'$ are the new respective intervals in running {$G(u_0)$} on input {$(\Vec{v}', h)$}, it holds that $B \subseteq B'$ and $D \subseteq D'$.
		See \Cref{fig:grimmett2} for an illustration of an example. 
		Thus, for any realization of the random shift {$u$} that awarded a given threshold of $t$ seats to the parties in $T$ when running $G(0)$ on {$(\Vec{v},h)$}, the same value of $u$ will award at least $t$ seats to the parties in $T$ when running {$G(u_0)$} on {$(\Vec{v}',h)$}, since the same $t$ integers that were included in $B$ and $D$ will be contained within $B'$ and $D'$. Since $G(0)$ and {$G(u_0)$} are both equivalent to Grimmett's method, we have shown that it satisfies threshold monotonicity.
	\end{proof}
	
	As discussed, the fact that systematic rounding violates selection monotonicity for coalitions of size $3$ implies that Grimmett's method violates threshold monotonicity for $|T|=3$. Our example in \Cref{fig:grimmett} also illustrates why the approach applied in the proof of \Cref{proGrimmett2} fails for larger coalitions: For the old interval of party $3$ to be included in its new interval, we would need to have a positive shift $u_0>0$. However, any positive shift would violate the same constraint for party $1$.

	\subsection{Impossibility of Pairwise Threshold Monotonicity}
	When we considered selection monotonicity and its pairwise variant in \cref{sec:rounding}, the most promising rounding rules seemed those whose probability distributions spread their mass smoothly and broadly over the possible outcomes (like conditional Poisson and Sampford rounding, and unlike systematic and pipage rounding).
	Furthermore, an easy argument showed that Sampford rounding satisfied the pairwise variant of selection monotonicity.

	In contrast, when moving to pairwise \emph{threshold} monotonicity, we show that no apportionment method whose distributions have such a full support, and thus neither the conditional Poisson method nor the Sampford method, can satisfy the axiom.

	\begin{restatable}{theorem}{thmImpossibilityFDCO}\label{thmImpossibilityFDCO}
		No (ex-ante proportional and quota) apportionment method that has full support can satisfy pairwise threshold monotonicity.
	\end{restatable}

	\begin{proof}
		Let $a$ be any proportional and optimistic apportionment algorithm. We fix $h = 3$ and consider a collection of vote vectors of size $n=12$ where all lower quotas are zero. Thus, we may think of the input to $a$ as just being the residue vector (which is the same as the vector of standard quotas).
		For any integer $m \geq 2$, we define the residues %
		$$\Vec{p}(m) \coloneqq \left(\frac15, \frac15, \frac15, \frac15, \frac15, \frac15, \frac15, \frac15, \frac15, \frac15, \frac1m, 1 - \frac1m \right)$$
		and let $S_m$ be the random variable defined as the set of parties that $a$ rounds up on input $p(m)$.
		By applying union bound over all events in which $S_m$ contains at least two elements from $\{1, \dots, 10\}$, we get
		\begin{align*}
			\sum_{R \in {\{1, \dots, 10\} \choose 2}}\Pr[S_m \supseteq R] &\geq \Pr[|S_m \cap \{1, \dots, 10\}| \geq 2] 
			= \Pr[|S_m \cap \{11, 12\}| \leq 1]\\[-7pt]
			&= 1 - \Pr[S_m \supseteq \{11, 12\}] \geq 1 - \Pr[11 \in S_m] \\ 
			& = 1 - \frac1m \geq \frac12.
		\end{align*}
		
		Therefore, by the averaging principle there must exist some $R \subseteq \{1, \dots, 10\}$ of size 2 such that
		$$\Pr[S_m \supseteq R] \geq \frac{1}{2 \cdot {10 \choose 2}} = \frac{1}{90}.$$
		
		Thus, for each $m \in \{2, 3, 4, \dots\}$, choose some $R_m \subseteq \{1, \dots, 10\}$ of size 2 such that
		$$\Pr[S_m \supseteq R_m] \geq \frac{1}{90}.$$
		Since the sequence $R_2, R_3, R_4, \dots$ only takes finitely-many different values, by the pigeonhole principle there must be some infinite subsequence of indices $m_1, m_2, m_3, \dots$ such that $R_{m_i}$ is the same for each positive integer $i$. Call this set $T_2$ and choose an arbitrary set $T_1 \subseteq \{1, \dots, 10\} \setminus T_2$ of size 3. Without loss of generality, we re-label the states so that $T_2 = \{1, 2\}$ and $T_1 = \{3, 4, 5\}$. Now consider the alternative residue vector $$\Vec{p}' \coloneqq \left(\frac14, \frac14, \frac16, \frac16, \frac16, \frac{1}{500}, \frac{1}{500}, \frac{1}{500}, \frac{1}{500}, \frac{1}{500}, 1 - \frac{1}{200}, 1 - \frac{1}{200} \right)$$
		and let $S'$ be the random set that $a$ rounds up on input $\Vec{p}'$. By full support, we know that
		$$\Pr[S' \supseteq T_1] \eqqcolon \varepsilon > 0.$$
		Choose $i$ large enough so that $\frac{1}{m_i} < \varepsilon$. We claim that the pair of residue vectors $(\Vec{p}', \Vec{p}(m_i))$ violates pairwise threshold monotonicity for the sets $T_1$ and $T_2$, with respective thresholds $\theta_1 = 3$ and $\theta_2 = 2$ (so we are interested in the probabilities that \emph{all} elements of $T_1$ or $T_2$ are rounded up).
		
		First consider the set $T_2$. Observe that, in passing from $\Vec{p}'$ to $\Vec{p}(m_i)$, residues within $T_2$ have decreased from $\frac14$ to $\frac15$. However, since selecting both states in $T_2$ entails excluding at least one of states 11 and 12, we have
		\begin{align*}
			\Pr[S' \supseteq T_2] &\leq \Pr[11 \notin S' \txt{ or } 12 \notin S'] \leq \Pr[11 \notin S'] + \Pr[12 \notin S'] \tag{by the union bound}\\
			&= 2 - \Pr[11 \in S'] - \Pr[12 \in S'] = 2 - \left(1 - \frac{1}{200}\right) - \left(1 - \frac{1}{200}\right) \tag{by proportionality}\\
			&= \frac{1}{100} < \frac{1}{90} \leq \Pr[S_{m_i} \supseteq R_{m_i}] = \Pr[S_{m_i} \supseteq T_2],
		\end{align*}
		so the probability of selecting all states in $T_2$ has increased.
		
		Next consider the set $T_1$. In passing from $\Vec{p}'$ to $\Vec{p}(m_i)$, residues in $T_1$ increased from $\frac16$ to $\frac15$, but
		\begin{align*}
			\Pr[S' \supseteq T_1] &= \varepsilon > \frac{1}{m_i} = 1 - \left(1 - \frac{1}{m_i}\right) = 1 - \Pr[12 \in S_{m_i}] \tag{by proportionality}\\
			&= \Pr[12 \notin S_{m_i}] \geq \Pr[S_{m_i} \supseteq T_1],
		\end{align*}
		where in the final inequality we have used the fact that selecting $T_1$ entails not selecting any other state, as $T_1$ has size 3. Thus, the probability of selecting all states in $T_1$ has decreased, and so we have shown that pairwise threshold monotonicity is violated.
	\end{proof}
	
	Though this theorem only rules out methods having full support, we see this as a consequence of the difficulty of making a uniform argument over many partially degenerate apportionment methods, not as an invitation to explore apportionment methods without full support.
	In particular, the failure of systematic and pipage rounding to satisfy pairwise selection monotonicity implies that their induced methods fail pairwise threshold monotonicity, as described in \cref{sec:apportionment:defs}.
	
	We also remark that our proof leans heavily on the fact that pairwise threshold monotonicity allows vote shares outside of the coalitions to both grow and shrink, which is why we see no indication that threshold monotonicity would be subject to a similar impossibility.
	
	\subsection{Impossibility of Monotonicity in Raw Vote Counts}
	\label{sec:apportionment:votecount}
	In the previous two axioms for apportionment, what entitled a coalition to more seats was if its constituent parties increased their \emph{share} of the votes.
	This seems especially natural in the party-apportionment setting\,---\,for example, when elections are covered in the media, the vote percentages are typically reported as the central statistics of interest.
	Alternatively, one could consider increases in the raw vote counts as what entitles a coalition to more joint representation.
	(There is precedent for both perspectives in classical apportionment theory, though this latter perspective is arguably more prominent~\cite[e.g.,][p.\ 106--108]{BY01}.)
	The two axioms coincide whenever the total population does not change. 
	When the overall number of votes shifts, however, axioms based on the vote counts may grant a coalition increases in joint representation in additional scenarios. %
	
	Unfortunately, formulating (pairwise) threshold monotonicity in terms of vote counts is too much to ask, i.e., the resulting axioms cannot be satisfied by any method. The proof is (as often in apportionment theory) based on the fact that, when we rescale the populations after a decrease in the total number of votes, this rescaling increases the standard quota of large parties by a larger absolute amount than the standard quota of small parties, which can bring the residue of such a large party close to one and force us to select this party almost always.
	
	\begin{restatable}{theorem}{thmimpossibilityvdc}\label{thmImpossibilityVDC}
		No (ex-ante proportional and quota) apportionment method can satisfy vote-count threshold monotonicity or pairwise vote-count threshold monotonicity.
	\end{restatable}
	
	\begin{proof}
		Let $a$ be any proportional apportionment method. We fix the house size be $h = 8$ and consider the initial vote vector
		$$\Vec{v} \coloneqq \left(380, 140, 140, 140\right).$$
		Let $S$ be the random variable defined as the set of parties that $a$ rounded up on input $(\Vec{v},h)$. Observe that the total number of votes is 800, and thus the lower quotas and residues are given by 
		$$\lfloor \Vec{q} \rfloor = (3, 1, 1, 1) \text{ and } \Vec{p} = (0.8, 0.4, 0.4, 0.4).$$
		Hence, $S$ must be of size 2. By proportionality,
		\begin{align*}
			0.8 &= \Pr[1 \in S] = 1 - \Pr[1 \notin S]\\
			&= 1 - \left(\Pr[S = \{2, 3\}] + \Pr[S = \{2, 4\}] + \Pr[S = \{3, 4\}]\right).
		\end{align*}
		Rearranging, we get
		$$\Pr[S = \{2, 3\}] + \Pr[S = \{2, 4\}] + \Pr[S = \{3, 4\}] = 0.2.$$
		Thus, by the averaging principle, at least one of these three probabilities must be at least $0.2/3 = 0.0\overline{6}$. Since $\Vec{v}$ is symmetric with respect to parties 2, 3, and 4, we may assume without loss of generality that this holds for the first term. That is, letting $T_1 \coloneqq \{2, 3\}$, we have
		$$\Pr{}[S = T_1] \geq 0.0\overline{6}.$$
		Now consider the alternative vote vector
		$$\Vec{v}' \coloneqq (376, 142, 142, 100).$$
		The total population is now 760, and the lower quotas remain unchanged. However, the residues are now
		$$\Vec{p}' = (0.96, 0.05, 0.49, 0.49)$$
		(rounded to two decimal places). Let $S'$ be the random set of size 2 that $a$ rounds up on input $(\Vec{v}', h)$. Since the residues still sum to 2, rounding up both parties in $T_1$ precludes rounding up party 1, so
		\begin{align*}
			\Pr[S' = T_1] &\leq \Pr[1 \notin S'] = 1 - \Pr[1 \in S']\\
			&= 1 - 0.96 = 0.04. \tag{by proportionality}
		\end{align*}
		We claim that $a$ violates pairwise vote-count threshold monotonicity for the pair of vote vectors $(\Vec{v},\Vec{v}')$, with $T_1 = \{2, 3\}$ and $T_2 \coloneqq \{1\}$, with thresholds $\theta_1 = \theta_2 = 4$.
		
		First consider $T_2$, which controls at least 4 seats if and only if it is rounded up, since the lower quota for party 1 (the only party in $T_2$) is 3. In passing from $\Vec{v}$ to $\Vec{v}'$, the number of votes for this party has decreased from 380 to 376, yet, by proportionality, the probability that $T_2$ is selected must have increased from 0.8 to 0.96.
		
		Next consider $T_1$, which controls at least 4 seats if and only if both of parties 2 and 3 are rounded up, since the lower quotas for these parties are each 1. In passing from $\Vec{v}$ to $\Vec{v}'$, the numbers of votes for both parties in $T_1$ have increased from 140 to 142, yet we have shown that
		$$\Pr[S = T_1] = 0.0\overline{6} > 0.04 \geq \Pr[S' \supseteq T_1],$$
		so the probability that $T_1$ controls 4 seats has decreased. Thus, $a$ violates pairwise vote-count threshold monotonicity. Since $v'_i \leq v_i$ for all $i \notin T_1$, this also shows that $a$ violates vote-count threshold monotonicity for $T_1$.
	\end{proof}
	
	Note that the proof requires $n = 4$ states. This is tight since for $n = 3$ there is only one proportional rounding rule to use, and it is easy to see that it satisfies all monotonicity properties.
	
	\section{Conclusion}

	From our results, a landscape of monotonicity notions in randomized apportionment and rounding is taking shape.
	On one side lie the (pairwise) selection monotonicity axioms, which can be attained by Sampford sampling.
	On the other side, we have the vote-count variants and pairwise threshold monotonicity, overly strong axioms that cannot be satisfied (at least, for the latter, by natural apportionment methods).
	The next step in the exploration of monotonicity will focus on the axiom in their middle: threshold monotonicity, which we conjecture to be feasible but could not yet show for general-size coalitions.
	
	One motivation to pursue threshold monotonicity is that it would imply a general notion of strategy-proofness for randomized apportionment.
	Consider a voter who approves a subset of the parties, in the sense that they want as many legislative seats as possible to be filled from this subset.\footnote{Such a utility is natural in a setting of proportional representation. For example, the approved parties could be those supporting a new policy the voter hopes to see adopted.}
	Suppose that the voter maximizes an expected utility that is a monotone function in the number of seats\,---\,a very general assumption, which can in particular capture the value of certain seat thresholds, as well as risk-averse or risk-seeking preferences. 
	A simple consequence of threshold monotonicity is that this voter can never profit from strategically voting for a disapproved party, and this even holds for coalitions of such voters.
	
	Beyond the setting of apportionment, we are broadly intrigued by the prospect of relational axioms~\cite{Thomson11} on higher-order correlations, i.e., axioms that capture how the higher-order correlations of a random process vary as a function of changes in the process's parameters.
	Going beyond rounding subject to cardinality constraints, where else can we obtain such guarantees, and where in algorithms, mechanism design, and statistics could they prove useful?
	
	\section{Acknowledgments}
	
	We are very grateful to Ariel Procaccia and Dominik Peters for many productive discussions and useful ideas.
	
	This work was partially supported by the National Science Foundation under Grant No.\txt{} DMS-1928930 and by the Alfred P.\txt{} Sloan Foundation under grant G-2021-16778, while three of the authors were in residence at the Simons Laufer Mathematical Sciences Institute (formerly MSRI) in Berkeley, California, during the Fall 2023 semester; and also by the National Science Foundation Graduate Research Fellowship Program under Grant No. DGE1745303. Any opinions, findings, and conclusions or recommendations expressed in this material are those of the authors and do not necessarily reflect the views of the National Science Foundation.
	
	This work was also partially supported by Anillo ICMD (ANID-Chile grant ACT210005), the Center for Mathematical
	Modeling (ANID-Chile grant FB210005), and FONDECYT ANID-Chile grant 1241846.
	
	\bibliographystyle{ACM-Reference-Format}
	\bibliography{bibliography,other}
	
	\clearpage
	\appendix
	\section*{\LARGE Appendix}
	\section{Relationship With Prior Work}\label{app:housewins}
	\begin{observation}
		Neither Grimmett's method, the pipage method, the conditional Poisson method, nor the Sampford method, nor any apportionment method satisfying full support satisfies the house monotonicity axiom of \citet{GPP22}.
	\end{observation}
	\begin{proof}
		As \citet{GPP22} show in Section~4.1 of their paper that, in a scenario with four parties, vote vector $v_1 = 1, v_2 = 2, v_3 = 1, v_4 = 2$, and house size $2$, no method satisfying quota and house monotonicity may give both seats to parties~1 and 3 with positive probability.
		By definition, any method satisfying full support would do so, and thus cannot satisfy house monotonicity, which rules out the conditional Poisson and the Sampford method.
		\citet{GPP22} show in the same section that Grimmett's method (with or without a randomized order) fails this test and thus fails house monotonicity.
		Finally, one easily verifies that the pipage method, when presented with the parties in the above worst-case order, also has a positive probability of awarding both seats to parties~1 and 3. When the order of parties is chosen uniformly at random, the order will with nonzero probability be the one above, showing again that the method fails house monotonicity.
	\end{proof}
	
	\section{Selection Monotonicity Counterexamples for Specific Rules}\label{appCounterexamples}

	Here we prove Propositions \ref{proPipageCounter} and \ref{proPoissonCounter} by way of counterexamples, showing that none of the rounding rules mentioned in this paper satisfy selection monotonicity except for Sampford sampling. We then compliment this result with Proposition \ref{proPairwiseCounter}, which states that these rules also all violate pairwise selection monotonicity.
	
	\proPipageCounter*
	\begin{proof}
		
		We start by showing the statement for pipage rounding. For the sake of illustration, we start with a counterexample for a fixed order corresponding to the order of the indices. Let $$\Vec{p} = \Big(\frac{1}{3},\frac{1}{2},\frac{1}{3},\frac{2}{3},\frac{2}{3},\frac{1}{2}\Big)$$ and $T = \{1,3,6\}.$ We claim that $$\Pr[S = T]>0.$$ To see this, consider the following update sequence which happens with positive probability by the definition of pipage rounding:

		\begin{table}[h]
			\centering
			\renewcommand{\arraystretch}{1.3}
			\begin{tabular}{c|cccccc}
				Step & \( p_1 \) & \( p_2 \) & \( p_3 \) & \( p_4 \) & \( p_5 \) & \( p_6 \) \\
				\hline
				1& \( \frac{1}{3} \) & \( \frac{1}{2} \) & \( \frac{1}{3} \) & \( \frac{2}{3} \) & \( \frac{2}{3} \) & \( \frac{1}{2} \) \\
				2& \( \frac{5}{6} \) & 0 & \( \frac{1}{3} \) & \( \frac{2}{3} \) & \( \frac{2}{3} \) & \( \frac{1}{2} \) \\
				3 &\( \frac{1}{6} \) & 0 & 1 & \( \frac{2}{3} \) & \( \frac{2}{3} \) & \( \frac{1}{2} \) \\
				4& \( \frac{5}{6} \) & 0 & 1 & 0 & \( \frac{2}{3} \) & \( \frac{1}{2} \) \\
				5& 1 & 0 & 1 & 0 & \( \frac{1}{2} \) & \( \frac{1}{2} \) \\
				6 & 1 & 0 & 1 & 0 & 0 & 1 \\
				
			\end{tabular}
		\end{table}
		Now, consider the new residue vector $$\Vec{p}' = \Big(\frac{1}{3},\frac{1}{3},\frac{1}{3},\frac{2}{3},\frac{2}{3},\frac{2}{3}\Big).$$
		We claim that for this profile $\Pr[\{1,3\} \subseteq S] = 0$. To see this, first consider the case that $p_1$ is set to $\frac{2}{3}$ in the first step. Then, in the second step, either $p_1$ is fixed to $1$ and $p_3$ to $0$ or vice versa. Otherwise, $p_1$ is set to $0$ already in the first step. Thus $\Pr[S = T] = 0$.    
		
		We also give the following counterexample that works for the version of pipage rounding in which we randomize over all possible orders. Let $T = \{2,3,4\}$ and $$\Vec{p} = \Big(0.07, 0.9, 0.57, 0.37, 0.99, 0.1\Big).$$ Computing the probabilities by a straight-forward brute-force implementation of the method yields $$\Pr[S = T] \geq 0.0043.$$  However, decreasing the residue for party 1 and increasing it for party 2 as follows
		$$\Vec{p}' = \Big( 0.06, 0.91, 0.57, 0.37, 0.99, 0.1\Big)$$
		yields $$\Pr[S = T] \leq 0.0042,$$ 
		showing that also the random order version of pipage rounding violates selection monotonicity.
		
		\bigskip 
		For systematic rounding with fixed order, the counterexample for Grimmett's apportionment method, given in \Cref{secIntro} directly provides a counterexample when using the residues given in \Cref{tbl:apportia} as input for systematic rounding. For random order, a straightforward brute-force check shows that in the first election there exists no order for which the left-wing parties get all three additional seats with positive probability. 
	\end{proof}
	
	\proPoissonCounter*
	\begin{proof}
		The following counterexample was discovered and verified with considerable computational effort, using symbolic and rational arithmetic to ensure there were no numerical errors. It involves such orders of magnitude that it is not detectable using standard Python packages with floating-point arithmetic. Consider the following pair of scaled working probabilities for $n = 6$ and $k = 3$:
		\begin{align*}
			\vec{\pi} &\coloneqq (99620001435175085845613951348591, 33206667145059699577734936400435, \push 33206667145059699577734936400435, 23244667001544291253373835102276586, \push 23244667001544291253373835102276586, 1660333357252963458777541885429371)\\
			\vec{\pi}' &\coloneqq (99620001435175193801835755646020, 33206667145059681577227243883092, \push 33206667145059681577227243883092, 23244667001544299141767505142336500, \push 23244667001544299141767505142336500, 1660333357252962147206216649823732)
		\end{align*}
		Note that, to get the true working probabilities, these must be scaled down by the following respective normalization constants, which do not fit on one line:
		\begin{align*}
			N = \txt{}& 999999999999999999999999999999997683569039326925979\\&611190826394109716067937415025227822500962454159264^{1/3}\\
			N' = \txt{}& 999999999999999999999999999999999031428737331175393\\&563500944194957519703813036589912171397358763602688^{1/3}
		\end{align*}
		One may verify that these correspond to a pair of residue vectors $\vec{p}$ and $\vec{p}'$ such that $p_1 \leq p'_1$, $p_2 \leq p'_2$, $p_3 \leq p'_3$, $p_4 \geq p'_4$, $p_5 \geq p'_5$, $p_6 \geq p'_6$, and yet
		$$\frac{\pi_1 \pi_2 \pi_3}{N} \geq \frac{\pi'_1 \pi'_2 \pi'_3}{N'},$$
		i.e., the set $\{1, 2, 3\}$ is more likely under $\vec{p}$ than $\vec{p}'$.
	\end{proof}

	\begin{proposition}\label{proPairwiseCounter}
		Systematic rounding, pipage rounding, and conditional Poisson rounding violate pairwise selection monotonicity. 
	\end{proposition}
	
	\begin{proof}
		For pipage rounding (with random order), consider the following vector of residues $$\Vec{p} = (0.03, 0.59,0.07,0.42,0.47,0.42),$$ summing up to $k=2$. We consider the two coalitions $T_1 = \{1,2\}$ and $T_2=\{5,6\}$. A straight-forward implementation of pipage rounding (with randomized order) yields the following selection probabilities for the two coalitions: 
		$$\Pr[S = T_1] \geq 0.011 \text{ and } \Pr[S = T_2] \leq 0.12.$$
		Now, consider the new vector of residues 
		$$\Vec{p}' = (0.03, 0.59,0.16,0.33,0.47,0.42).$$ Note that, while residues for $T_1$ and $T_2$ remained unchanged, the residues of parties $3$ and $4$ did change. This leads to coalition $T_1$ losing selection probability and coalition $T_2$ gaining selection probability. Precisely, $$\Pr[S = T_1] \leq 0.01 \text{ and } \Pr[S = T_2] \geq 0.1204,$$ which proves the violation of pairwise selection monotonicity. 
		
		\bigskip
		
		For systematic rounding, we let $n = 6$ and $k = 2$, and suppose $\vec{p} = (0.1, 0.1, 0.2, 0.2, 0.5, 0.9)$. Let $T_1$ be the set of parties with residue $0.1$ and $T_2$ be the set of parties with residue $0.2$. Then one can check that, if we increase $0.5 \to 0.6$ and decrease $0.9 \to 0.8$ (without changing the total population or state quotas), the probability of selecting $T_1$ decreases (to zero) while the probability of selecting $T_2$ increases, which violates pairwise selection monotonicity.
		
		\bigskip
		
		Finally, for conditional Poisson rounding, we must resort to computational verification as in the proof of Proposition \ref{proPoissonCounter}. We unfortunately do not have a closed form for the working probabilities, but the pair of residue profiles is
		\begin{align*}
			\vec{p} \coloneqq (0.0618342562928861,\ & 0.0207176796116814, \\ 0.0207176796116814,\ & 0.9933997806603289, \\ 0.9933997806603289,\ & 0.9099308231630933)\\
			\vec{p}' \coloneqq (0.0618342562928862,\ & 0.0207176796116814, \\ 0.0207176796116814,\ & 0.9933997806603289, \\ 0.9933997806603289,\ & 0.9099308231630932)
		\end{align*}
		Observe that the only residues that changed in going from $\vec{p}$ to $\vec{p}'$ were $p_1$, which increased, and $p_6$, which decreased. The unique working probabilities are algebraic numbers given by complicated expressions. Using Mathematica, we computed that the set $T_1 \coloneqq \{1, 2, 3\}$ (which contains only residues that weakly increased) decreased in probability by roughly $5 \times 10^{-26}$ while $T_2 \coloneqq \{3, 4, 5\}$ (which contains only residues that weakly decreased, because they all stayed the same) increased in probability by roughly $2 \times 10^{-18}$. Thus, pairwise selection monotonicity is violated.
	\end{proof}

	\section{Sufficient Condition for Selection Monotonicity} \label{app:Sufficient}
	We call a rounding rule neutral, if reordering of the input vector $\Vec{p}$ does not change the selection probabilities. Clearly, Sampford rounding is neutral. Under neutrality, we can assume without loss of generality that the set of parties of interest in the definition of selection monotonicity is $T=[k]$. 
	
	\begin{lemma}
		Let $r$ be a neutral rounding rule. If, for any $\Vec{p} \in [0,1)^{n}$ with $\sum_{1 \leq i \leq n} p_i = k$ for some integer $1 \leq k \leq n-1$ and $p_1 > 0$ it holds that 
		\[\left(\pdif{p_1} - \pdif{p_n}\right) \, \mathbb{P}_{S \sim r(\Vec{p})}[S = [k]] \geq 0,\]
		then $r$ satisfies selection monotonicity. 
	\end{lemma}
	
	\begin{proof}
		Let $\Vec{p}, \Vec{p}' \in [0,1)^n$ be two vectors summing up to the same integer $k$. Let $T$ be a set of $k$ parties such that $p'_i \geq p_i$ for all parties $i \in T$ and $p'_i\leq p_i$ for all $i \not\in T$. By neutrality we can reorder $\Vec{p}$ and $\Vec{p}'$ such that $T = [k]$ and do not change the selection probabilities of $T$. To prove that $r$ satisfies selection monotonicity, we aim to show that $$\Pr_{S \sim r(\Vec{p}')}[S = [k]] \geq \Pr_{S \sim r(\Vec{p})}[S = [k]].$$
		We prove the statement by induction over the number of indices in $[n]$ for which $p_i \neq p'_i$ holds. Since $\Vec{p}$ and $\Vec{p}'$ are normalized, they differ in at least two indices (unless they are equal). Thus, in our base case, we assume without loss of generality that $p'_1 = p_1 + \epsilon$ and $p'_n = p_n - \epsilon$ for some $\epsilon \in (0,1)$, where the specific form again follows from normalization. We define the following function, which linearly interpolates between $\Vec{p}$ and $\Vec{p}'$. That is, for all $\lambda \in [0,1]$ we define $$\Vec{f}(\lambda) = \Vec{p} + \lambda (\Vec{p}' - \Vec{p}).$$ Note that $f_1(\lambda) = p_1 + \lambda\epsilon $ and $f_n(\lambda) = p_n - \lambda \epsilon$ and $f_i(\lambda) = p_i = p'_i$ for all $i \in \{2, \dots, n-1\}$. We get that 
		\begin{align}
			\frac{\partial}{\partial \lambda} \Pr_{S \sim r(\Vec{f}(\lambda))}[S = [k]] &= \frac{\partial}{\partial f_1}  \Pr_{S \sim r(\Vec{f})}[S = [k]]\frac{\partial f_1}{\partial \lambda} + \frac{\partial}{\partial f_n} \Pr_{S \sim r(\Vec{f})}[S = [k]]\frac{\partial f_n}{\partial \lambda} \tag{chain rule}\\ 
			& = \epsilon  \cdot \left( \frac{\partial}{\partial f_1}\Pr_{S \sim r(\Vec{f})}[S = [k]] -  \frac{\partial}{\partial f_n}\cdot \Pr_{S \sim r(\Vec{f})}[S = [k]] \right) \notag \\ 
			& \geq 0 \tag{precondition of lemma}
		\end{align}
		Thus, by the fundamental theorem of calculus and the fact that $\Vec{f}(0) = \Vec{p}$ and $\Vec{f}(1) = \Vec{p}'$, we get that 
		$$\Pr_{S \sim r(\Vec{p}')}[S = [k]] - \Pr_{S \sim r(\Vec{p})}[S = [k]] = \int_{0}^1 \frac{\partial}{\partial \lambda} \Pr_{S \sim r(\Vec{f}(\lambda))}[S = [k]] d\lambda \geq 0.$$
		For the induction step, consider the case when $\Vec{p}$ and $\Vec{p}'$ differ in more than two indices. Let $i$ be the index minimizing $|p'_i - p_i|$ among all $i \in [n]$ with $p'_i \neq p_i$ and assume without loss of generality that $p'_i > p_i$. Then, let $j \in [n]$ be such that $p_j < p'_j$. Let $\Vec{w}$ be the vector with $w_\ell = p_\ell$ for all $\ell \in [n] \setminus \{i,j\}$, $w_{i} = p'_i$ and $w_{j} = p_j - (p'_i - p_i)$. Note that $\Vec{w} \in [0,1)^n$, $\Vec{w}$ sums to $k$ and $\Vec{p}$ and $\Vec{w}$ differ in two indices, namely $w_i >p_i$ and $w_j < p_j$, where $i \in [k]$ and $j \in [n] \setminus [k]$ holds. Hence, using the same argument as in the base case, we can show that $$\Pr_{S \sim r(\Vec{w})}[S = [k]] \geq \Pr_{S \sim r(\Vec{p})}[S = [k]].$$
		Moreover, note that $\Vec{w}$ and $\Vec{p}'$ differ in one less index than $\Vec{p}$ and $\Vec{p}'$, and it also holds that $p'_\ell \geq w_{\ell}$ for all $\ell \in [k]$ and $p'_{\ell} \leq w_{\ell}$ for all $\ell \in [n] \setminus [k]$. Thus, by induction hypothesis it follows that $$\Pr_{S \sim r(\Vec{p}')}[S = [k]] \geq \Pr_{S \sim r(\Vec{w})}[S = [k]] \geq \Pr_{S \sim r(\Vec{p})}[S = [k]],$$
		which concludes the proof. 
	\end{proof}
	
	\section{Deferred Proofs}
	\label{app:deferred}
	
	\subsection{Formula from \Cref{sec:rounding}} \label{app:deferred:rounding}

	\lemxpdif*
	\begin{proof}
		We take the derivative for each $\ell \in [k-1]$ individually. That is, 
		\begin{align*}
			&\left(\pdif{p_1} - \pdif{p_n}\right) \, \expbone \\
			={} &\sum_{0 \leq \ell \leq k-1} (k - \ell) \cdot \left(\pdif{p_1} - 
			\pdif{p_n}\right) \, \mathbb{P}\big[|B| = \ell\big] \\
			={} &\sum_{0 \leq \ell \leq k-1} (k - \ell) \cdot \left(\pdif{p_1} - \pdif{p_n}\right) \, \bigg((1-p_1)\,(1-p_n)\,\mathbb{P}\big[|\hat{B}| = \ell\big] \\
			&\hphantom{\sum_{0 \leq \ell \leq k-1} (k - \ell) \cdot \left(\pdif{p_1} - \pdif{p_n}\right) \, \bigg(}+ \big(p_1\,(1-p_n) + (1-p_1)\,p_n\big)\,\mathbb{P}\big[|\hat{B}| = \ell-1\big] \\
			&\hphantom{\sum_{0 \leq \ell \leq k-1} (k - \ell) \cdot \left(\pdif{p_1} - \pdif{p_n}\right) \, \bigg(}+ p_1\,p_n\,\mathbb{P}\big[|\hat{B}| = \ell-2\big]\bigg) \\
			={} &\sum_{0 \leq \ell \leq k-1} (k - \ell) \, \big((p_n - p_1)\,\mathbb{P}\big[|\hat{B}| = \ell\big] + 2\,(p_1 - p_n)\,\mathbb{P}\big[|\hat{B}| = \ell-1\big] + (p_n - p_1)\,\mathbb{P}\big[|\hat{B}| = \ell-2\big]\big) \\
			={} & (p_n - p_1) \, \bigg( \sum_{0 \leq \ell \leq k-1} (k - \ell) \, \big(\mathbb{P}\big[|\hat{B}| = \ell\big]- \mathbb{P}\big[|\hat{B}| = \ell-1\big]\big) \\
			& \hphantom{(p_n - p_1) \, \bigg(} - \sum_{0 \leq \ell \leq k-1} (k - \ell) \,\big(\mathbb{P}\big[|\hat{B}| = \ell - 1\big]- \mathbb{P}\big[|\hat{B}| = \ell-2\big]\big)\bigg) \\
			={}& (p_n - p_1) \, \left( \sum_{0 \leq \ell \leq k-1} \mathbb{P}\big[|\hat{B}| = \ell\big] - \sum_{0 \leq \ell \leq k-2} \mathbb{P}\big[|\hat{B}| = \ell\big]\right)\tag{telescoping sums, $\mathbb{P}\big[|\hat{B}| < 0\big]=0$} \\
			={}& (p_n - p_1) \, \mathbb{P}\big[|\hat{B}| = k-1\big],  
		\end{align*}
		which concludes the proof. \qedhere
	\end{proof}

	\subsection{Auxiliary Lemma for the Proof of \Cref{proGrimmett2}} \label{app:GrimmettTech}
	
	\begin{lemma}\label{lemGrimmett2Tech}
		Let $a,b,c,d,e \in \rr$ be such that $a, c, e \leq 0$, $b, d \geq 0$, and $a + b + c + d + e = 0$. Then there exists $u_0 \in \rr$ such that the following all hold:
		\begin{enumerate}[leftmargin=.9cm]
			\item\label{itmGrimmet2Tech1} $u_0 + a \leq 0$
			\item\label{itmGrimmet2Tech2} $u_0 + a + b \geq 0$
			\item\label{itmGrimmet2Tech3} $u_0 + a + b + c \leq 0$
			\item\label{itmGrimmet2Tech4} $u_0 + a + b + c + d \geq 0$
		\end{enumerate}
	\end{lemma}
	
	\begin{proof}
		There are two cases to consider. First suppose $b \leq -c$. Then we let $u_0 \coloneqq -a$, which satisfies all four inequalities:
		\begin{enumerate}[leftmargin=.9cm]
			\item $u_0 + a = 0$
			\item $u_0 + a + b = b \geq 0$
			\item $u_0 + a + b + c = b + c \leq -c + c = 0$
			\item $u_0 + a + b + c + d = b + c + d = -a - e \geq 0$
		\end{enumerate}
		On the other hand, if $b \geq -c$, we let $u_0 \coloneqq -a - b - c$, which then also satisfies all four inequalities:
		\begin{enumerate}[leftmargin=.9cm]
			\item $u_0 + a = - b - c \leq - b + b = 0$
			\item $u_0 + a + b = -c \geq 0$
			\item $u_0 + a + b + c = 0$
			\item $u_0 + a + b + c + d = d \geq 0$ 
		\end{enumerate}
		This concludes the proof. 
	\end{proof}
	
	By letting $a \coloneqq \abs{\abs{A'} - \abs{A}}$, $b \coloneqq \abs{\abs{B'} - \abs{B}}$, and so on, this lemma implies that $B \subseteq B'$ and $D \subseteq D'$ in the proof of Theorem \ref{proGrimmett2}.
\end{document}